\documentclass[12pt]{article}
\usepackage[top=1in, bottom=1in, left=1in, right=1in]{geometry}
\usepackage{pgfplots}
\pgfplotsset{compat=1.18}
\usepackage{microtype}
\usepackage{newtxtext}

\usepackage{latexsym} 
\usepackage{amsmath, amsfonts, amssymb, amsthm}
\usepackage{graphicx}
\usepackage{indentfirst}         
\usepackage{placeins}
\usepackage{graphics}
\usepackage{lscape}
\usepackage{setspace}
\usepackage{caption}
\usepackage{subcaption}
\usepackage{pdflscape}
\usepackage{dsfont}
\usepackage{float}
\usepackage{soul}
\usepackage{color}
\usepackage{url}
\usepackage{fancyhdr}
\usepackage{multirow}
\usepackage{booktabs}
\usepackage{threeparttable}
\usepackage{ragged2e}
\usepackage{bbm}
\usepackage[normalem]{ulem}
\usepackage{comment}
\usepackage{adjustbox}
\usepackage{multirow}
\usepackage{bm}
\usepackage{lmodern}
\usepackage{longtable}
\usepackage{rotating}
\usepackage{amsmath}
\usepackage{amsthm}
\usepackage{amssymb}
\usepackage{bm}
\usepackage[shortlabels]{enumitem}
\usepackage{placeins}
\usepackage{xcolor}
\usepackage{bbm}
\usepackage{pgfgantt}  
\usepackage{mathtools}      

\newtheorem*{assumption*}{Assumption}
\newtheorem{claim}{Claim}[section]
\newtheorem*{claim*}{Claim}
\usepackage{caption} 
\newenvironment{appfigure}{%
  \par\addvspace{18pt}\noindent\begin{minipage}{\textwidth}\centering
}{%
  \par\end{minipage}\par\addvspace{18pt}
}
\newenvironment{apptable}{%
  \par\noindent\begin{minipage}{\textwidth}\centering\tableformat
}{%
  \par\end{minipage}\par\clearpage
}

\usepackage{afterpage}
\usepackage{pdfpages}
\usepackage{hyperref}
\hypersetup{colorlinks,allcolors=black}
\hypersetup{citecolor=black,bookmarksdepth=2}
\hypersetup{linkcolor=black}
\usepackage[all]{hypcap}

\usepackage{setspace}
\usepackage[noabbrev]{cleveref}

\usepackage{etoolbox}

\newcommand{\tableformat}{%
  \small\setstretch{1}\setlength{\tabcolsep}{6pt}%
  \renewcommand{\arraystretch}{1.05}%
}

\definecolor{econblue}{HTML}{0072B2}
\definecolor{econgold}{HTML}{E69F00}
\newcommand{\wagepanel}[7]{%
  \begingroup
  \pgfmathsetlengthmacro{\wageplotwidth}{\linewidth-34pt}%
  \pgfmathsetlengthmacro{\wagexlabel}{34pt+0.5*\wageplotwidth}%
  \begin{tikzpicture}[font=\sffamily\fontsize{8}{9.5}\selectfont]
    \path[use as bounding box] (0,0) rectangle (\linewidth,162pt);
    \node[anchor=south west,inner sep=0] at (34pt,34pt)
      {\includegraphics[trim={#2},clip,width=\wageplotwidth,height=125pt]{#1}};
    \begin{scope}[shift={(34pt,34pt)},x=\wageplotwidth,y=125pt]
      \foreach \pos/\tick in {#3}
        \node[anchor=north,inner sep=0,yshift=-3pt] at (\pos,0) {\tick};
      \foreach \pos/\tick in {#4}
        \node[anchor=east,inner sep=0,xshift=-3pt] at (0,\pos) {\tick};
    \end{scope}
    \node[rotate=90,inner sep=0] at (6pt,96.5pt) {#5};
    \node[inner sep=0] at (\wagexlabel,16pt) {#6};
    \node[anchor=south,inner sep=0,font=\sffamily\fontsize{7}{8}\selectfont]
      at (\wagexlabel,0) {#7};
  \end{tikzpicture}%
  \endgroup
}

\newenvironment{figurebody}[1]{%
  \par\noindent\begin{minipage}{#1}\centering
}{%
  \par\end{minipage}\par
}
\newenvironment{floatnotes}[1][\RaggedRight]{%
  \par\addvspace{6pt}\nointerlineskip%
  \noindent\begin{minipage}{\linewidth}%
  \normalfont\footnotesize\setstretch{1}#1
  \setlength{\parindent}{0pt}\setlength{\parskip}{0pt}%
  \textit{Notes:}\space\ignorespaces
}{%
  \par\end{minipage}\par
}

\AtBeginEnvironment{abstract}{\onehalfspacing} 
\AfterEndEnvironment{abstract}{\onehalfspacing} 

\usepackage[backend=biber,style=authoryear,sorting=nyt,
            natbib=true,maxbibnames=99,maxcitenames=2,
            dashed=true,doi=false,isbn=false,url=false,date=year]{biblatex}
\usepackage{xurl}

\newcommand{\authoryearnp}[1]{\citealt{#1}}

\newcommand{\authoryearpos}[1]{%
  \citeauthor{#1}'s~\begingroup
    \let\NAT@open\relax \let\NAT@close\relax
    \citeyearpar{#1}%
  \endgroup
}

\newcommand{\yearbracket}[1]{%
  \begingroup \let\NAT@open=[ \let\NAT@close=] \citeyearpar{#1} \endgroup
}

\usepackage{titlesec}

\titleformat*{\section}{\large\bfseries}
\titleformat*{\subsection}{\large\bfseries}
\titleformat*{\subsubsection}{\large\bfseries}
\titlespacing*{\section}{0pt}{18pt plus 3pt minus 2pt}{8pt}
\titlespacing*{\subsection}{0pt}{14pt plus 2pt minus 2pt}{6pt}
\titlespacing*{\subsubsection}{0pt}{12pt plus 2pt minus 2pt}{5pt}

\usepackage{array}
\newcolumntype{P}[1]{>{\centering\arraybackslash}p{#1}}
\usepackage{adjustbox}

\title{Signaling in the Age of AI: Evidence from Cover Letters}
\author{Jingyi Cui, Gabriel Dias, and Justin Ye\footnote{Cui: UC Berkeley. Email: jingyi.cui@berkeley.edu. Dias: Yale University. Email: gabriel.santamarina@yale.edu. Ye: Yale University. Email: justin.ye@yale.edu. We are grateful to Elliot Ash for suggesting the research question and Freelancer.com for providing the data. Freelancer reviewed the draft for the sole purpose of identifying ``company data or confidential information contained in the publication which company does not
consent to disclose" and has no influence over the findings of the paper. We thank Joe Altonji, Steve Berry, Guillermo Carranza, Winnie van Dijk, Phil Haile, Charles Hodgson, Katja Seim, participants at Columbia MAD and Asian Conference on OrgEcon, and our editor and referees for helpful comments. We build on \citet{silbert} who first measured similarity between cover letters and job posts on Freelancer.com, interpreted it as a signal of worker ability, and showed it predicts hiring. This work was supported by funding from the Cowles Foundation for Research in Economics at Yale University.}}

\date{\today}

\begin{document}
\hypersetup{pageanchor=false}
\pagenumbering{arabic}
\setcounter{page}{1}

\maketitle
\thispagestyle{empty}

\vspace{0.25cm}
\begin{center}

\end{center}

\vspace{1cm}

\begin{abstract}

We study the impact of generative AI on labor market signaling using the introduction of an AI-powered cover letter writing tool to a subset of workers on a large online labor platform. Our data track
both access to the tool and usage at the application level. Difference-in-differences estimates
show that access to the tool increased textual tailoring between cover letters and job posts. After the tool's introduction, the correlation between cover letter tailoring and callbacks fell by 52\%. The decline is immediate and larger for experienced employers. In response, employers shifted toward alternative signals, namely workers’ work histories.


\par\vspace{0.5\baselineskip}
\end{abstract}

\newpage


\newpage
\pagenumbering{arabic}

\setcounter{page}{1}

\hypersetup{pageanchor=true}
\section{Introduction}

Job seekers and college applicants increasingly turn to generative artificial intelligence (AI) tools like ChatGPT to edit résumés, craft cover letters, and prepare for interviews.\footnote{This trend has received widespread media attention. Recent headlines include “AI is supposed to make applying to jobs easier—but it might be creating another problem” (NBC News, November 17, 2024) and “AI Is Taking Over College Admissions” (The Nation, October 11, 2024). In 2025, LinkedIn reported a 45 percent year-over-year increase in job applications, with generative AI widely believed to be a key driver of the surge (“Employers Are Buried in A.I.-Generated Résumés,” The New York Times, June 21, 2025).} Despite the transformative potential of this technology for processes central to economic opportunities and social mobility, credible empirical evidence on how AI affects labor market signaling remains scarce.\footnote{Notable exceptions include \citet{cowgill2024does}, who provide evidence on applicants' use of generative AI in a lab setting, and \citet{wiles2025generative}, who study the effect of employers' use of generative AI tools to write job posts on a similar online labor platform. } The main obstacle is observability: outside of controlled lab settings, just as employers struggle to identify AI-assisted applications, so too do researchers. Current detection tools are error-prone and unreliable \citep{weber2023testing}, making it difficult to assess AI assistance’s real-world impact on hiring.

We overcome this challenge by studying the introduction of a generative AI cover letter writing tool on Freelancer.com, one of the world's largest online labor platforms. Freelancer connects international workers and employers to collaborate on short-term, skilled, and mostly remote jobs. On April 19, 2023, Freelancer introduced the ``AI Bid Writer," a tool that automatically generates cover letters tailored to employers' job descriptions that workers can use or edit. The tool was available to a large subset of workers depending on their membership plans. Employers, as is typical in most real-world settings, could not observe which workers had access to the tool or whether a given cover letter was AI-assisted. 

We use eight months of data from two leading skill categories, PHP and Internet Marketing, covering 5 million cover letters submitted to over 100,000 jobs.\footnote{PHP, or Hypertext Preprocessor, is a widely used open-source scripting language designed for web development, enabling the creation of dynamic websites and applications.} At the application level, we observe access to and use of the tool and link this information to submitted cover letters, callbacks, and hiring outcomes. We also observe timestamps of workers’ clicks on the AI tool and application submissions, allowing us to measure time spent editing AI-generated text and examine human-AI interaction. We measure cover letter tailoring using cosine similarity of Term Frequency–Inverse Document Frequency (TF-IDF) vectors, a standard approach that measures how closely two texts align based on keyword importance. Our approach is related to, but distinct from, \citet{silbert}, who use data from the same platform and ask a Large Language Model (LLM) whether a given cover letter aligns with a job posting. They show that the LLM’s estimated probability of a “yes” response predicts hiring outcomes. Using data from before the introduction of the AI writing tool, we show that our TF-IDF measure similarly predicts callbacks and awards, both across and within workers, providing complementary evidence that employers value cover letters that engage with job content.

The platform-native tool was adopted by a meaningful share of workers. Following its launch, 64.6\% of eligible users used it at least once, and 16.8\% of all cover letters were written with its assistance. Since the tool directly affected a sizable share of applications, its rollout allows us to examine both the individual-level effects of generative AI and its market-level consequences, including changes in the informativeness of cover letters as a signal.

We first examine the effect of gaining access to the generative AI tool on cover letter tailoring and workers' chances of receiving a callback. We use a difference-in-differences design comparing workers with and without access. Identification is complicated by the release of GPT-4, a general-purpose generative AI technology, one month before the platform-native tool’s launch. We assume that the effect of GPT-4 stabilized after one month and attribute differential trends following the platform-native tool’s introduction to the tool itself. We provide evidence supporting this assumption. Access to the tool increased cover letter tailoring by 0.16 standard deviations. Sensitivity analyses following \citet{rambachan2023more} show that the estimate is robust to violations of the parallel-trends assumption. Instrumenting actual usage with access, we find that usage increased tailoring by 1.31 standard deviations. Applying the same design to callbacks, we find that access increased the probability of receiving a callback by 5.2\%. The estimate is significant at the 5\% level but not at the 1\% level and is not robust. These results show that AI substantially improves cover letter tailoring and provide suggestive evidence that AI helps workers receive callbacks.

We next find that AI substitutes for, rather than complements, workers' pre-AI cover letter tailoring skills. We compute workers’ average tailoring scores before the tool’s introduction and study heterogeneous effects by pre-existing writing ability. Workers who previously wrote more tailored cover letters experienced smaller gains in tailoring. By enabling less skilled writers to produce more tailored cover letters, AI can narrow differences in cover letter quality across workers and reduce the informativeness of cover letters as a signal. 

Consistent with this mechanism, we show that the correlation between cover letter tailoring and receiving a callback fell by 51.8\% after the tool’s launch, while the correlation with receiving an offer fell below zero. The decline in the relationship between tailoring and callbacks occurred for both cover letters generated with the tool and for those not generated with the tool. It emerged in the first month after the rollout, and persisted throughout our observation period. It was also larger for experienced employers—those who had hired a worker in the same skill categories during the three months before the tool’s introduction—than for inexperienced employers. One possible explanation is that experienced employers’ familiarity with the applicant pool made them more aware of the change in the information environment.

As employers placed less weight on cover letter tailoring, they shifted toward other signals less susceptible to AI influence, particularly workers’ reputation. The correlation between callbacks and workers’ reputation scores—the platform’s proprietary metric summarizing past work experience on the platform and determining the default ranking of applications—rose by 22\%. These patterns suggest that as AI adoption increases, employers may substitute away from easily manipulated signals such as cover letters toward harder-to-fake indicators of quality. This shift may have important implications, including favoring established workers over newcomers trying to enter the market. Despite these changes in signaling, we find no evidence of changes in aggregate interviewing, hiring, or job completion or in wage bids or realized wages following the tool’s introduction. We view this as a preliminary result that may not capture the long-term effect of AI.

Finally, we examine human-AI interaction. Although the tool automatically generates cover letters tailored to job posts, workers can edit the drafts before submitting them. Using timestamps of workers’ clicks on the AI tool and subsequent application submissions, we construct a measure of editing time. Most AI-generated cover letters were submitted with little or no modification, but workers with higher pre-AI writing ability spent more time editing AI drafts. Controlling for worker fixed effects, we find that a one-standard-deviation increase in editing time is associated with a 0.29 percentage point increase in the probability of receiving a job offer, equal to about 66\% of the average offer rate among AI-assisted applications. These results are consistent with human revision improving the effectiveness of AI-generated cover letters and contribute to a growing body of evidence on human-AI collaboration.

\paragraph{Related Literature} \citet{silbert} study similar questions about AI and labor-market signaling on Freelancer.com, including the rollout of the same platform-native AI writing tool. Using an alternative, LLM-based measure of the alignment between cover letters and job posts, they document a similar decline in the signaling value of customized cover letters following the diffusion of generative AI. Their main quantitative analysis then develops and estimates a structural model of costly signaling to study a counterfactual in which AI eliminates the cost of producing informative cover letters for all workers.

Our paper is complementary but focuses more directly on the realized effects and mechanisms of the AI tool's introduction. First, we exploit the fact that access to the tool was limited to workers on particular membership plans, allowing us to use a difference-in-differences design to estimate how gaining access affected workers' cover letters and hiring outcomes. Second, we provide evidence on employer adaptation to the new information environment, showing that more experienced employers respond more strongly to the changing informativeness of cover letters and that employers shift attention toward workers' reputation. Third, we directly examine whether the rollout affected other market outcomes, including bids, realized wages, matching, and job posting, and find limited changes along these margins. Finally, because we observe workers' interactions with the platform's AI tool, we study workers' time spent revising AI-generated drafts and how such human-AI interaction relates to hiring outcomes. Thus, while \citet{silbert} use a structural model to quantify the equilibrium consequences of a breakdown in costly signaling, we use variation surrounding the actual rollout to characterize its causal effects, the mechanisms through which employers and workers respond, and its realized effects across a broader set of market outcomes.

A few other recent papers study AI and labor matching. \citet{frankel2019muddled} provide a theoretical analysis of how gaming can worsen market information. \citet{cowgill2024does} show experimentally that generative AI can reduce screening accuracy. Our setting complements this work by studying realized hiring behavior when employers generally cannot observe AI use. In a related labor market, \citet{wiles2023algorithmic} show that non-generative writing assistance improves workers' employment outcomes without reducing employer satisfaction. \citet{wiles2025generative} study generative AI assistance for employers writing job posts and find more postings but no increase in matches. Together, these papers complement our focus on worker-side generative AI and its market-level effects.

Our paper also relates to work on online labor markets and platform-mediated employment \citep{horton2010online,chen2016are,stanton2020gig,stanton2025who,horton2024sorting,allman2025signaling,brinatti2023international}. We extend this literature by studying how generative AI changes the informativeness of signals used in matching. More broadly, we complement work on AI, productivity, and innovation \citep{agrawal2019artificial,noy2023experimental,brynjolfsson2025generative,acemoglu2025simple} by emphasizing labor-market signaling. Recent work also shows that generative AI can reshape worker behavior, employer decisions, and the distribution of opportunities \citep{dillon2025shifting,cao2025fairness,cao2025aiwrites,chen2025mediocrity}. Our focus is on how AI changes workers' ability to signal match quality through cover letters. This mechanism also relates to work on human--AI interaction \citep{agarwal2023combining,hosanagar2024designing,chakrabarty2025can}.

Finally, in investigating what happens when one signal of worker quality becomes less informative, we are related to the literature on statistical discrimination \citep{agan2018ban,doleac2020unintended,friedberg2021benefits,bartik2025deleting,corbae2025employer}.

\paragraph{Organization} Section \ref{sec: context} describes the market and the rollout of the AI tool. Section \ref{sec: data} presents the data, our measure of cover letter tailoring, and tool adoption patterns. Section \ref{sec:pe} studies the individual-level impact of generative AI on cover letters and hiring outcomes. Section \ref{sec: ge} turns to market-level consequences, including those on the informativeness of cover letters. Section \ref{sec: human_ai} examines the human-AI interaction. Section \ref{sec: conclusion} concludes. 

\section{\label{sec: context} Context}

In this section, we describe the platform and the generative AI–powered cover-letter writing tool that constitute the focus of our analysis.

\subsection{Freelancer.com}

Our analysis employs data from Freelancer.com, which ranks among the world's largest online task-based labor markets both in terms of user base and the volume of posted jobs. Platforms like Freelancer.com, Upwork, and Fiverr serve as digital marketplaces that connect workers and employers worldwide, facilitating the exchange of remotely delivered tasks such as software development, sales and marketing support, and creative work—commonly referred to as “freelancing.” Since its inception in 2009, the platform has attracted over 83 million users from 247 countries, regions, and territories, with over \$1 billion in gross payments transacted. 

On Freelancer.com, the matching process unfolds in five main steps. First, firms or individual employers seeking to outsource a task post a project—also referred to as a ‘job’ in this paper—on the platform. The employer describes the job and specifies a minimum budget, which is the lowest amount a worker can bid. Most jobs are fixed-budget rather than hourly. We interpret the minimum budget as a measure of the job’s size, and in all subsequent analyses we normalize workers’ wage bids by this minimum budget. Second, workers with skills matching the job requirements submit bids on these jobs. Bidding for jobs on Freelancer.com is free for workers; however, the number of bids a worker can submit each month is limited. Workers can increase their monthly bidding allowance by subscribing to one of the platform’s membership plans (see https://www.freelancer.com/membership/). Each bid typically comprises two components: the proposed payment amount and a written job proposal, which serves as a cover letter outlining the worker’s suitability for the task. Third, based on the set of submitted bids and the characteristics of the competing workers, the platform employs a recommendation algorithm (as a function of the volume and strength of a bidder's past reviews, referred to henceforth as the ``reputation score") to rank the bids for employers. Fourth, employers review the bids, displayed in order of ranking, and select the most appealing one, initiating the work process. Alternatively, employers may opt to leave the job unfilled if none of the bids meet their requirements. During this process, employers can and often do reach out to a subset of bidders for a private conversation, which we refer to as a callback. Finally, after the task is completed, both the worker and the employer rate each other’s performance by assigning a star rating, which contributes to their respective reputations on the platform. Payments are handled through the platform and are typically divided into milestone payments.

\subsection{The AI Bid Writer}

On April 19, 2023, Freelancer.com launched the AI Bid Writer. With a single click, the tool automatically generates a cover letter tailored to the job description using information on the worker, thereby reducing the effort and time required for freelancers to bid. While the auto-generated cover letters are ready for immediate submission, freelancers are free to edit the text before submission.

The release of the AI Bid Writer was unannounced beforehand, ensuring that its adoption was unaffected by pre-release anticipation. The tool is available exclusively to users subscribed to the Plus or higher membership tiers. The Plus plan costs \$8.95 per month and permits up to 100 bids, compared with 6 bids for free members and 50 bids for Basic members, who pay \$4.99 per month. Using the AI Bid Writer incurs no additional cost.\footnote{Further information about the feature and its functionality is provided on Freelancer.com’s support page: AI Bid Writer Overview (https://www.freelancer.com/support/project/ai-bid-writer).}

In our data, we observe workers’ membership plans—and thus whether they had access to the AI Bid Writer—as well as whether they clicked on the tool’s button for each bid submitted. Importantly, this information is not available to employers: they neither know workers’ membership plans nor observe whether a cover letter was generated using the AI Bid Writer. This feature is essential for the interpretation of our results: cover letters that were customized with and without AI could appear similar to the employers, leading employers to reconsider the informational value contained in well-tailored cover letters. This inability to detect the use of AI tools is also a common feature in other real-world contexts, such as college applications.

The AI Bid Writer tool was not the only generative AI tool that workers had access to in the sample period. One alternative is GPT-4, a general-purpose AI tool that OpenAI released on March 14, 2023, just over a month before the launch of the AI Bid Writer. Compared to its predecessor GPT-3.5 (released prior to our sample period on November 30, 2022), GPT-4 exhibits substantially enhanced reasoning. GPT‑4’s gains in logical coherence, nuanced query interpretation, and adherence to complex, long‑form instructions are well-documented in benchmark evaluations (\authoryearnp{OpenAI2023GPT4}). Beyond raw reasoning, GPT-4 also shows advances on writing tasks, including cover-letter generation, and exhibits emergent theory-of-mind behavior that may support more audience-aware drafting (\authoryearnp{Deveci2023CoverLetters}; \authoryearnp{Sikander2023Intro}; \authoryearnp{Kosinski2023}).

We believe that the platform-native tool was comparable to GPT-4 in terms of the quality of the cover letters it produced. The main differences lay in cost and ease of use: while the AI Bid Writer was free for workers on qualifying membership plans, access to GPT-4 required a \$20 monthly subscription. In addition, the AI Bid Writer automatically reads the job posting and generates a cover letter with a single click, whereas freelancers using GPT-4 would need to copy and paste the job description into the OpenAI interface or develop a script to automate the process—potentially incurring additional costs through the ChatGPT API. Due to these advantages of the AI Bid Writer, the tool was used by a substantial set of workers (which we show in Section \ref{sec: adoption}) despite the concurrent availability of other tools like GPT-4. 

In our analyses, we account for the GPT-4 launch when estimating the effects of the platform-native AI writing tool. We discuss the assumption in Section \ref{section: DiD_tailoring}.



\FloatBarrier

\section{\label{sec: data} Data Patterns}

This section begins by presenting the data, sample restrictions, and summary statistics. We then describe our measure of cover letter tailoring and demonstrate that it predicts hiring outcomes. Next, we document adoption patterns of the AI Bid Writer tool and examine whether its launch coincided with changes in membership plan purchases or job postings.

\subsection{Data Description and Sample Restrictions}
Our analysis focuses on all bids submitted to jobs requiring skills categorized under ``PHP" and ``Internet Marketing" (two of the most popular skill categories) during the 240-day period between January 19, 2023, and September 15, 2023. We exclude bids where the proposed wage exceeds 30 times the minimum job budget, as these outliers likely do not reflect typical bidding behavior. Finally, we omit jobs that received only a single bid, as these are likely direct offers.

The analytic sample comprises 5,800,993 bids submitted by 282,789 unique bidders across a total of 114,244 jobs. Wage bids average 3.339 times the minimum job budget, reflecting a range of pricing strategies. At the bidder level, the average number of bids submitted during the sample period is 20.5, though this figure masks substantial heterogeneity. A minority of highly active bidders account for the majority of bids, reflecting a skewed distribution of activity. On average, bidders receive 0.11 awards in PHP and Internet Marketing jobs during the sample period. At the job level, the average job receives 50.8 bids, underscoring the competitive nature of the platform.



\subsection{TF-IDF: a Measure of Cover Letter Tailoring}

Because the AI Bid Writer automatically produces cover letters customized to job posts, it dramatically reduces the costs of tailoring applications. To measure cover letter tailoring, we focus on textual alignment, which captures the degree of overlap in keywords between a cover letter and a job description. The underlying assumption is that more tailored cover letters should exhibit greater textual similarity with the job posting, reflecting the freelancer’s effort in addressing the specific requirements of the employer. On the contrary, generic cover letters which are re-used across jobs would exhibit lower textual similarity. 

To quantify this, we utilize Term Frequency-Inverse Document Frequency (TF-IDF), a simple text vectorization technique that assigns weights to words based on their importance within a document while accounting for their overall prevalence across the corpus. Specifically, following a pre-processing step to remove stopwords, we construct TFIDF vectors for each job post and cover letter. Each element in the TFIDF vector represents the weighted frequency of a word within the text. The TFIDF score for a given word \( w \) in document \( d \) is computed as:
\begin{align*}
    \text{TFIDF}(w, d) = \text{TF}(w, d) \times \text{IDF}(w)
\end{align*}
where $\text{TF}(w, d)$ represents the term frequency of word $w$ in document $d$, and $ \text{IDF}(w)$ is the inverse document frequency, given by:
\begin{align*}
    \text{IDF}(w) = \log \left( \frac{N}{1 + n_w} \right)
\end{align*}
where \( N \) is the total number of documents in the corpus, and \( n_w \) is the number of documents containing word \( w \).

After computing the TFIDF vectors for both the job description and the bid proposal, we measure their cosine similarity to obtain a numerical representation of textual alignment. The cosine similarity between two TFIDF vectors, \( v_p \) (for the job description) and \( v_b \) (for the cover letter), is defined as:
\begin{equation}
    \text{Cosine Similarity} = \frac{v_p \cdot v_b}{\| v_p \| \| v_b \|}
\end{equation}
This measure ranges from \( 0 \) to \( 1 \), where a value closer to \( 1 \) indicates a higher degree of textual similarity between the job description and the cover letter, while a value near \( 0 \) suggests little to no overlap.

\begin{table}[!htbp]
\centering\tableformat

\caption{Cover Letter Tailoring and Bid Success}
\label{tab:tfidf_valid}
\begin{threeparttable}
\begin{tabular*}{0.90\textwidth}{@{\extracolsep{\fill}}lcc} 
\toprule 
& \multicolumn{2}{c}{Offer} \\ 
 & (1)  & (2)  \\ 
\midrule
Cover Letter Tailoring  & 0.0075***  & 0.0041*** \\
& (0.0005) & (0.0006) \\
Wage Bid & -0.0007***  & -0.0005*** \\
& (0.00002) & (0.00003) \\
Worker Reputation & 0.0119*** & -0.0001 \\ 
& (0.0002) & (0.0008) \\ 
\midrule
Job FE & \checkmark & \checkmark \\
Bidder FE & & \checkmark \\
Outcome Mean & 0.0059 & 0.0059 \\ 
Outcome Std. Dev & 0.0766 & 0.0766\\
\midrule
\textit{N} & 2,085,655 & 2,018,690 \\
\bottomrule
\end{tabular*}
\begin{floatnotes}
The table reports regression results where the dependent variable is an indicator for whether the bid is successful. Column (1) reports estimates controlling for job fixed effects, while Column (2) reports estimates additionally controlling for bidder fixed effects. Standard errors are reported in parentheses and clustered at the job level. Cover letter tailoring is measured as textual similarity—the cosine similarity between TF-IDF vectors of a job proposal and a job description. Wage is normalized as the wage bid divided by the minimum job budget. Reputation is the ranking percentile (0-1) of a platform-calculated reputation score representing a worker's past work performance on the platform and reviews among bids for the job. *** p $<$ 0.01, ** p $<$ 0.05, * p $<$ 0.1.
\end{floatnotes}
\end{threeparttable}

\end{table}

By constructing this textual alignment measure, we obtain a quantitative metric of how closely a freelancer’s cover letter mirrors the employer’s job posting, allowing for an analysis of how cover letter quality evolves over time and in response to AI. This measure, while simple, has several advantages. First, it is highly computationally inexpensive. Second, it is largely insensitive to word count and is normalized by document length—hence, variation in the word counts of job descriptions or cover letters does not introduce bias. Third, it is insensitive to word order or sentence structure, and hence explicitly focuses on content similarity while ignoring stylistic variations.\footnote{\citet{wiles2023algorithmic} study the impact of an AI tool that primarily corrects spelling and grammar errors in a similar context.} 

We validate that in the pre-rollout sample, our measure of cover letter tailoring is significantly predictive of hiring. In Table \ref{tab:tfidf_valid}, we regress the outcome of whether a bid is selected on cover letter tailoring, the worker's reputation (captured by the reputation score), and the normalized wage (the bid amount divided by the minimum job budget). Accounting for job fixed effects, we find that a one-standard-deviation increase in cover letter tailoring (0.134) corresponds to an increase in absolute selection probability of 0.10 percentage points, or a 17.02\% increase in the probability of being selected compared to the raw mean. This positive correlation is robust to the inclusion of worker fixed effects. To put this magnitude in context, a one-standard-deviation increase in normalized wage ($3.203$) is associated with a $0.21$ percentage point decrease in the probability of being hired, while a one-standard-deviation improvement in rank percentile ($0.289$) corresponds to a $0.34$ percentage point increase in the probability of selection. Together, these results support the interpretation of cover letter tailoring as a labor-market signal—namely, that firms use it to form beliefs about otherwise unobservable worker attributes, such as productivity, interest in the job, or match-specific quality.

Appendix Table \ref{tab:similarity_scores} includes examples of cover letter tailoring for selected bids on a job. While our measure of cover letter tailoring captures the dimension most directly affected by the AI Bid Writer, we acknowledge that cover letters also vary along other important dimensions—such as style—that our measure does not capture. Assessing the impact of generative AI tools on these additional aspects of cover letters lies beyond the scope of this paper and is left for future work.

\FloatBarrier

\subsection{Tool Adoption Patterns \label{sec: adoption}}

\begin{figure}[!htbp]
\centering
\begin{figurebody}{0.82\textwidth}
\caption{Distribution of Worker-Level AI Tool Adoption Rates Among Workers with Access}
\label{fig:histogram-adoption}

  \centering 
  
  \includegraphics[width=\linewidth]{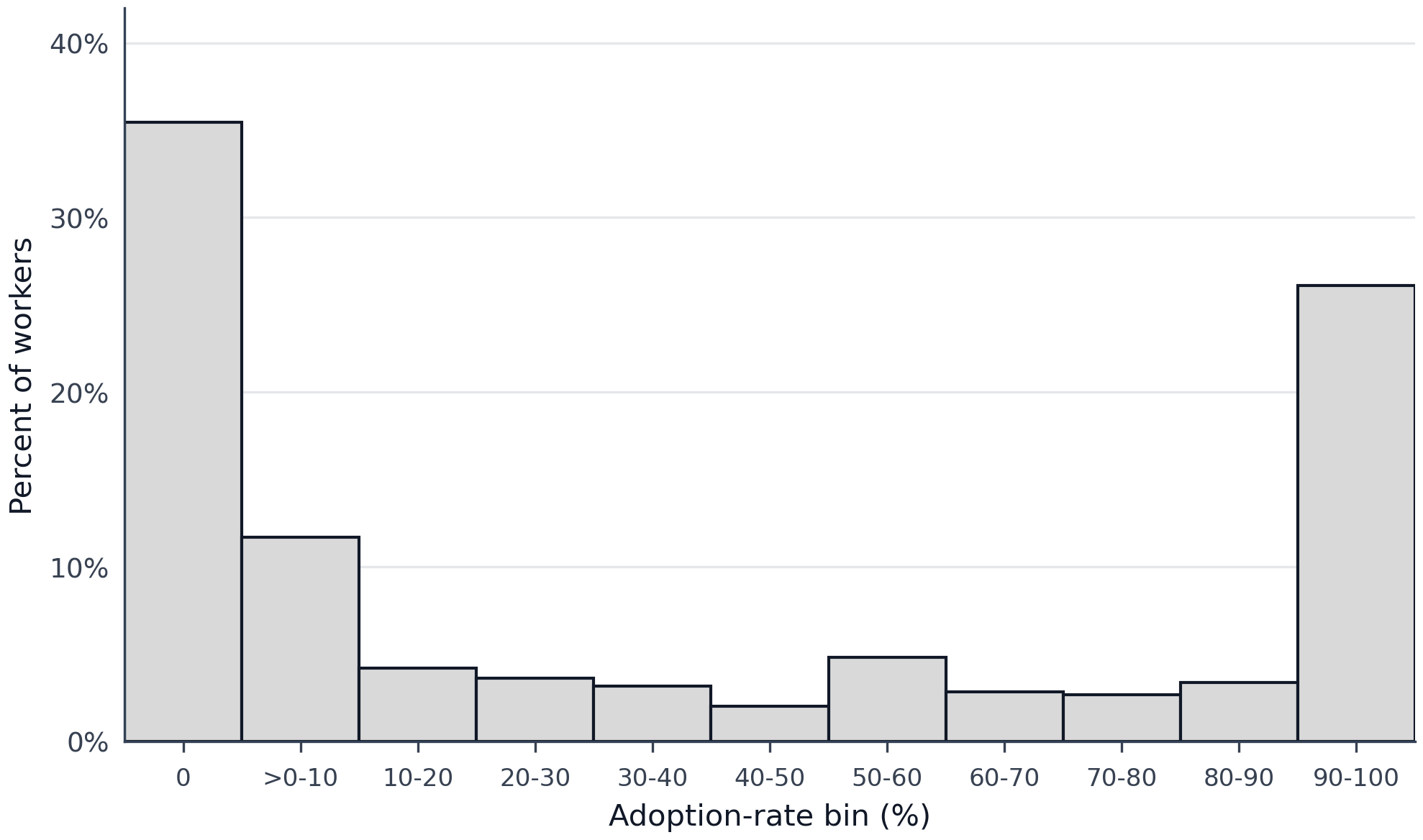}

    \begin{floatnotes}
Adoption rates are measured at the worker level as the fraction of bids submitted with the AI tool during the period when it was available. The sample is restricted to workers who had access to the tool.
\end{floatnotes}
  
\end{figurebody}
\end{figure}

Using data on workers’ membership plans and per-application usage of the AI Bid Writer, we analyze adoption of the tool. We find that 64.6\% of workers with access tried it at least once. Figure~\ref{fig:histogram-adoption} presents the histogram of adoption rates among these users. The distribution is bimodal: about 47.1\% of workers used the tool in fewer than 10\% of their bids, while more than 25.9\% relied on it in over 90\% of their bids. This wide dispersion underscores the heterogeneous nature of AI adoption, with workers pursuing different strategies regarding when and how frequently to incorporate AI into their cover letters.

In the post-AI sample, about 18.8\% of eligible bids—those submitted by workers with access to the tool—were written with assistance from the AI Bid Writer. Since these workers account for the vast majority of bids in our data (85\%), 16.8\% of all bids were AI-assisted. This substantial adoption provides the basis for analyzing the market-level effects of AI in Section \ref{sec: ge}.

\FloatBarrier

\subsection{Trends in User Memberships and Job Postings}

In this section, we examine whether the launch of the AI Bid Writer coincided with changes in membership plan purchases or in the number of projects posted.

The AI cover-letter writing tool was available to workers on Plus and higher membership tiers. A potential concern is that workers adjusted their membership choices in response to the feature. Figure \ref{fig:freq_notbasic} shows that purchases and renewals of AI-eligible plans remained stable around the tool’s introduction, as did their share of all membership purchases and renewals over the full sample. This suggests that the tool did not materially affect selection into membership tiers, consistent with the AI feature being only one of many benefits of higher-tier plans, the most prominent of which was a larger monthly bid allowance.


    
   

\begin{figure}[!htbp]
\centering
\begin{figurebody}{\textwidth}
\caption{Membership Purchases}
\label{fig:freq_notbasic}


    \includegraphics[width=\linewidth]{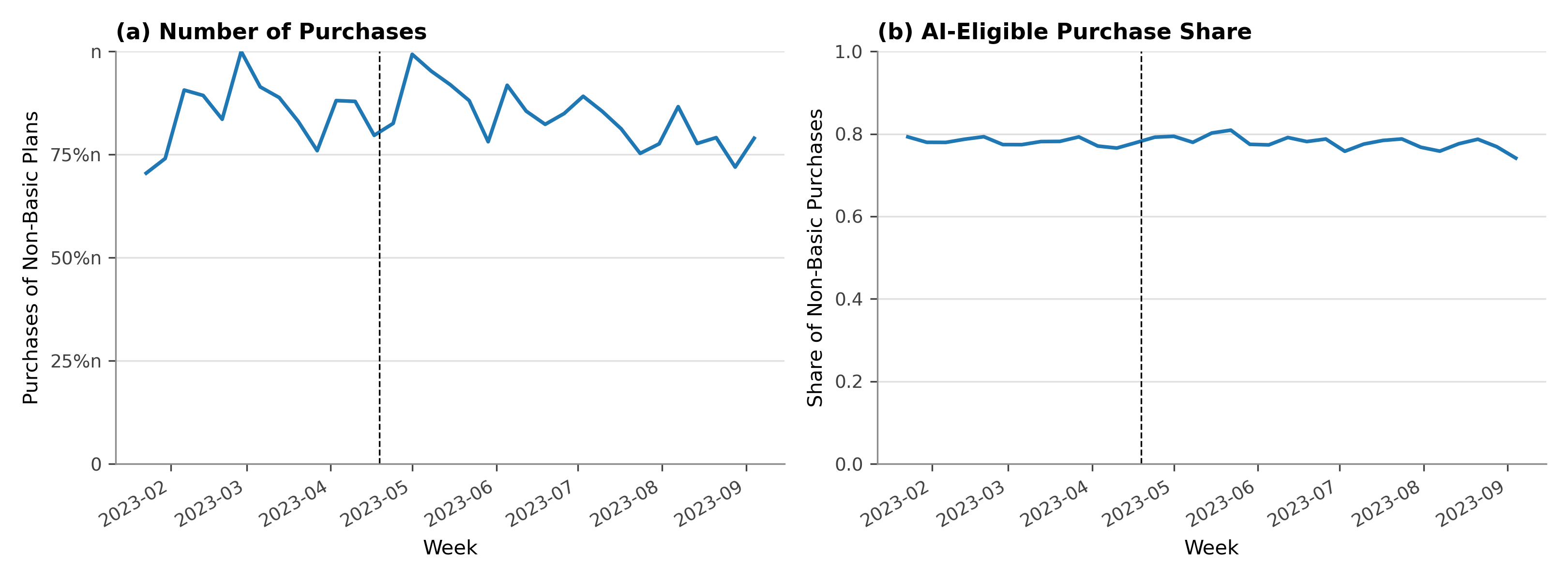}

    \begin{floatnotes}
The left panel reports the number of purchases and renewals of membership plans that provide access to the AI tool. The right panel reports the share of such plans relative to all membership plan purchases and renewals. Due to data confidentiality, the y-axis in panel (a) is normalized. The maximum value is shown as ‘n’, and other ticks represent fractions of n.
\end{floatnotes}
   
\end{figurebody}
\end{figure}

Another potential concern is that the rollout of the AI cover-letter tool might have coincided with changes in job postings. Figure~\ref{fig:jobs-week}, however, shows that the weekly number of new jobs in our sample remained stable around the launch.

\begin{figure}[!htbp]
\centering
\begin{figurebody}{0.75\textwidth}
\caption{Weekly Number of Jobs Posted}
\label{fig:jobs-week}

\centering

  \includegraphics[width=\linewidth]{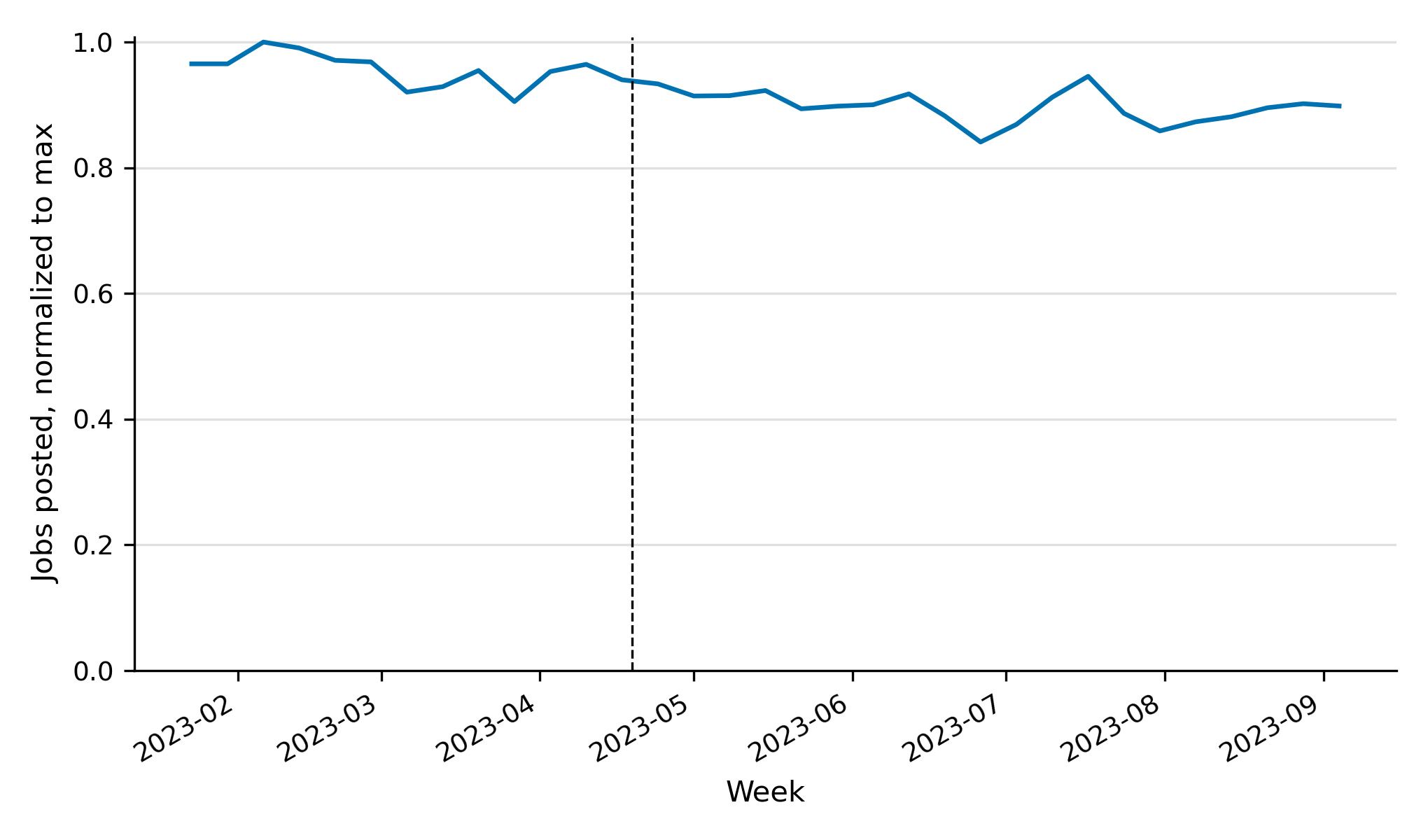}

   \begin{floatnotes}
A job's week is defined as the week when its first bid was submitted. Due to data confidentiality, the y-axis is normalized. The maximum value is shown as 1, and other ticks represent fractions of the max.
\end{floatnotes}
\end{figurebody}
\end{figure}

\FloatBarrier

\section{\label{sec:pe} The Individual-Level Impact of Generative AI}

In this section, we study the individual-level effects of the generative AI writing tool. Using a two-way fixed effect specification, we compare the outcomes between workers who did and did not have access to the tool before and after the rollout of the tool. The outcomes we consider are cover letter tailoring and the probability of receiving a callback.

\subsection{The Impact of AI on Cover Letter Tailoring \label{section: DiD_tailoring}}

\begin{figure}[!htbp]
\centering
\begin{figurebody}{0.82\textwidth}
\caption{Time Series of Average Cover Letter Tailoring by Treatment Status}
\label{fig:avg_tfidf}

  \centering   
  \includegraphics[width=\linewidth]{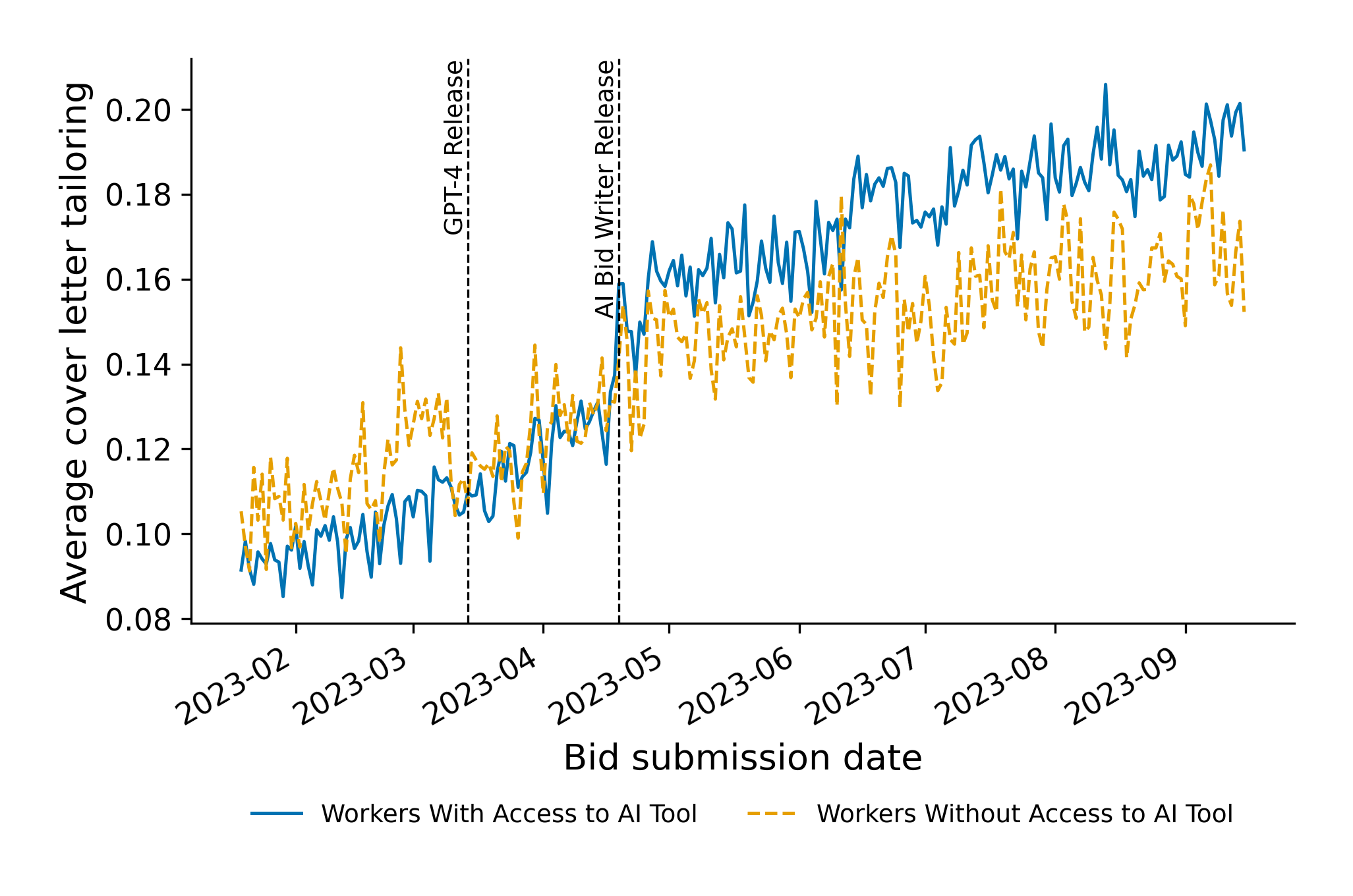}

    \begin{floatnotes}
Cover letter tailoring is computed as the cosine similarity between the TF-IDF vectors of the cover letter and the job description.
\end{floatnotes}  
\end{figurebody}
\end{figure}

For the individual-impact analysis, we rely on a difference-in-differences specification examining the within-worker change in outcomes between workers who had access to the AI tool and those who did not. Consequently, we omit 1,347,872 bids from 259,101 individuals who were not active both before and after the intervention, and 2,013,376 bids from 5,440 individuals whose access to the AI tool changed during the sample period. These restrictions yield a final sample of 2,439,745 bids from 18,248 workers. Appendix Table \ref{tab:samplecomp} compares observables in the trimmed and full samples. Relative to the full sample, the trimmed estimation sample is composed of more active bidders: at the bidder level, trimmed-sample workers are active for more weeks, submit more bids per active week, and have accumulated more reputation on the platform. Other outcomes are similar. 

As a descriptive check on the treatment–control comparison underlying the design, Table \ref{tab:preperiod_balance} reports pre-treatment observables for treated and control workers in the trimmed estimation sample. Treated workers are substantially more active on the platform, consistent with treatment being determined by membership tiers that differ in bid allowances. They also have higher reputation. At the same time, the two groups are similar in pre-period tailoring of cover letters, normalized wage bids, callback rates, and hiring rates, suggesting that baseline application quality and hiring success did not differ substantially across groups prior to the rollout of the tool. 

\begin{table}[!htbp]
\centering\tableformat
\caption{Comparison of Pre-Period Observables between Treatment and Control Workers}
\label{tab:preperiod_balance}
\begin{threeparttable}
\begin{tabular*}{0.90\textwidth}{@{\extracolsep{\fill}}lcccc}
\toprule
Variable & Control mean & Treated mean & Diff.\ $(T-C)$ & SMD \\
\midrule
Weeks Active & 2.32 & 7.19 & 4.87 & 1.36 \\
Bids per Active Week & 2.08 & 19.2 & 17.2 & 0.587 \\
Cover Letter Tailoring & 0.112 & 0.122 & 0.00969 & 0.0891 \\
Normalized Wage Bid & 3.04 & 3.20 & 0.168 & 0.0898 \\
Reputation Score & 0.594 & 0.772 & 0.178 & 1.51 \\
Callback Rate & 0.0721 & 0.0860 & 0.0139 & 0.0848 \\
Award Rate & 0.00752 & 0.0132 & 0.00573 & 0.0968 \\
\midrule
Bidders & 14{,}395 & 3{,}853 & \multicolumn{2}{l}{Total = 18{,}248} \\
\bottomrule
\end{tabular*}
\begin{floatnotes}
This table reports pre-treatment bidder-level averages for workers included in the DiD estimation sample, separately for control and treated workers. Reputation score is a platform-calculated score summarizing workers' work histories on the platform. Normalized wage bid is the bid amount normalized by the job's minimum budget. Diff.\ (T-C) denotes the treated mean minus the control mean. SMD denotes the standardized mean difference. The statistics are specific to our data sample and are not platform-wide. 
\end{floatnotes}
\end{threeparttable}
\end{table}

We begin by examining the impact of access to the AI writing tool on cover letter tailoring. Figure \ref{fig:avg_tfidf} presents a visual comparison of average cover letter tailoring between bidders with and without access to the AI tool. Before the rollout of the AI tool and the rollout of GPT-4, the tailoring of cover letters written by workers with access followed a trend similar to that for workers without access. However, after the rollout, the two lines diverged, with the former exhibiting a noticeable and sustained increase in tailoring.

To formally estimate the effect, we employ a two-way fixed effects intention-to-treat (ITT) estimator, specified as follows:
\[
y_{ijt} = \delta_i + \gamma_t + \beta \text{PostAI}_t\times\text{Access}_i+ \lambda \text{PostGPT}_t\times\text{Access}_i +  x_{ijt}+\epsilon_{ijt}
\]
where the outcome $y_{ijt}$ is cover letter tailoring for the bid submitted by worker $i$ at auction $j$ in month $t$, $\delta_i$ captures worker fixed effects, $\gamma_t$ represents month fixed effects, $\text{PostGPT}_t$ is an indicator for whether the bid was submitted after the release of GPT-4 (March 14), $\text{PostAI}_t$ is an indicator for whether the bid was submitted post platform tool rollout (April 19), $\text{Access}_i$ is an indicator for whether the bidder had access to the tool at rollout, and $x_{ijt}$ is a vector of additional controls including the reputation level of the worker at the time of bidding (based on their platform-calculated reputation score), the normalized wage bid (wage bid divided by the minimum job budget), and the log number of bids for the job.

In a standard difference-in-differences design, identification comes from comparing two groups (treatment and control) before and after treatment implementation. However, pre-treatment exposure to an analogous external technology (GPT-4) can confound this comparison in our setting. To address this, we impose the following identification assumption: although the effect of GPT-4 may differ between treated and control workers, its impact stabilizes within one month of its introduction. In other words, we assume that, in the absence of the platform-native tool, control and treatment workers would have followed parallel trends a month after GPT-4's introduction. 

While the above specification recovers an intention-to-treat (ITT) effect of AI access on labor market outcomes, we additionally obtain the local average treatment effect (LATE). The need to estimate LATE arises from imperfect compliance: not all workers with access to the AI writing tool used it, and those who did use it did not necessarily do so for every auction. This non-compliance implies that the ITT estimate is a lower bound of the effect of AI usage, as it is an average over cases where the eligible workers used the tool and cases where the eligible workers did not use the tool. 
Appendix Table~\ref{tab:complier_characteristics} shows that, among workers with access, AI adopters had substantially greater pre-period bidding activity than nonadopters, but similar cover letter tailoring and award rates. 

To estimate the LATE, we instrument for AI usage with treatment using the below specification:
\[
y_{ijt} = \delta_i + \gamma_t + \beta \text{AI Bid}_{ijt} + x_{ijt} + \lambda \text{PostGPT}_t\times\text{Access}_i + \epsilon_{ijt},
\]
where the dependent variable, \(y_{ijt}\), again represents cover letter tailoring for bidder \(i\) on auction \(j\) in month \(t\). The variable \(\text{AI Bid}_{ijt}\) is an indicator for whether the AI tool was used to write the bid. Observing usage of the AI tool at the application level is a key advantage of our paper. We instrument for \(\text{AI Bid}_{ijt}\) using $\text{PostAI}_t \times \text{Access}_i$, the interaction of AI tool access and the post-rollout period.

\begin{table}[!htbp]
\centering\tableformat
\caption{Effect of AI on Cover Letter Tailoring}
\label{tab:tfidf_regression}
\begin{threeparttable}
\begin{tabular*}{0.90\textwidth}{@{\extracolsep{\fill}}lcccc}
\toprule
 & \multicolumn{4}{c}{Cover Letter Tailoring} \\
 & (1) ITT & (2) LATE & (3) ITT & (4) LATE \\
\midrule
AI & 0.0256*** & 0.2017*** & 0.0209*** & 0.1765*** \\
   & (0.0032) & (0.0190) & (0.0032) & (0.0200) \\
\midrule
Bidder FE & \checkmark & \checkmark & \checkmark & \checkmark \\
Month FE & \checkmark & \checkmark & \checkmark & \checkmark \\
$\text{PostGPT}_t \times \text{Access}_{i}$ &  &  & \checkmark & \checkmark \\
Outcome Mean & 0.1508 & 0.1508 & 0.1508 & 0.1508 \\
Outcome Std. Dev. & 0.1639 & 0.1639 & 0.1639 & 0.1639 \\
\midrule
\textit{N} & 2,439,745 & 2,439,745 & 2,439,745 & 2,439,745 \\
\bottomrule
\end{tabular*}
\begin{floatnotes}
Cover letter tailoring is measured as textual similarity--the cosine similarity between TF-IDF vectors of a cover letter and a job description. Columns (1) and (3) report the ITT (the coefficient of PostAI $\times$ Access in the two-way fixed effect specification), while columns (2) and (4) report the LATE, where we instrument for the usage of AI using PostAI $\times$ Access. Columns (1) and (2) omit the PostGPT $\times$ Access term; columns (3) and (4) include it. Wage bids, reputation score rank percentile, and log bid count are controlled for. Standard errors are clustered at the bidder level. *** p $<$ 0.01, ** p $<$ 0.05, * p $<$ 0.1.
\end{floatnotes}
\end{threeparttable}
\end{table}

Table \ref{tab:tfidf_regression} contains the results from the above specifications, where columns (1) and (2) exclude the $\text{PostGPT}_t\times\text{Access}_i$ term. We find that across specifications, access to the AI tool significantly increased cover letter tailoring. Examining column (3), using the ITT specification, the coefficient on the interaction term is 0.0209, which is 15.6\% of the pre-rollout cover letter tailoring standard deviation of 0.135. Relative to the pre-rollout baseline average cover letter tailoring of 0.109 among those with access, this corresponds to an overall increase of 19.2\%. 

Because compliance is incomplete, our ITT estimates of accessing the tool provide only a lower bound on the effect of using it. 
In the LATE specification in column (4), the coefficient on using AI bids is 0.1765, which is 131\% of the pre-rollout standard deviation in cover letter tailoring of 0.135. Relative to the baseline mean cover letter tailoring among the treated, these estimates suggest that usage of AI Bid Writer among compliers is associated with a 161\% increase in cover letter tailoring. 

To assess the robustness of the ITT results to violations of the parallel trends assumption, we conduct two complementary sensitivity analyses following \citet{rambachan2023more}, reported in Appendix Figure \ref{fig:honestdid_tfidf}. The first uses our monthly event-study specification in Equation \ref{eq: tfidf_dynamic}. We obtain a breakdown value of 1.38, indicating that the estimated positive effect of access to the AI tool on cover letter tailoring in the first month ($\beta_0$) remains statistically significant even when post-treatment deviations from parallel trends are allowed to be up to 38\% larger than the largest pre-treatment deviation. In the second exercise, we estimate the event study at the weekly level and restrict the sample to the period after the introduction of GPT-4, using the four weeks between the introduction of GPT-4 and the rollout of the platform-native AI tool as the pre-treatment period. The resulting breakdown value is larger, at 1.67.

To evaluate dynamic effects, we additionally perform the following ITT specification:
\begin{equation}
\label{eq: tfidf_dynamic}
y_{ijt} = \delta_i + \gamma_t + \sum_{k \neq -1} \beta_k \, \mathbf{1}\{t = k\} \times \text{Access}_i \;\;+ \text{PostGPT}_t \times \text{Access}_i + x_{ijt} + \; \epsilon_{ijt}.
\end{equation}
Figure \ref{fig:tfidf_did} reports the $\beta_k$ coefficients, indicating an immediate and persistent effect associated with having access to the tool.

\begin{figure}[!htbp]
\centering
\begin{figurebody}{0.82\textwidth}
\caption{Dynamic Effects of AI on Cover Letter Tailoring (ITT)}
\label{fig:tfidf_did}

  \centering
  \includegraphics[width=\linewidth]{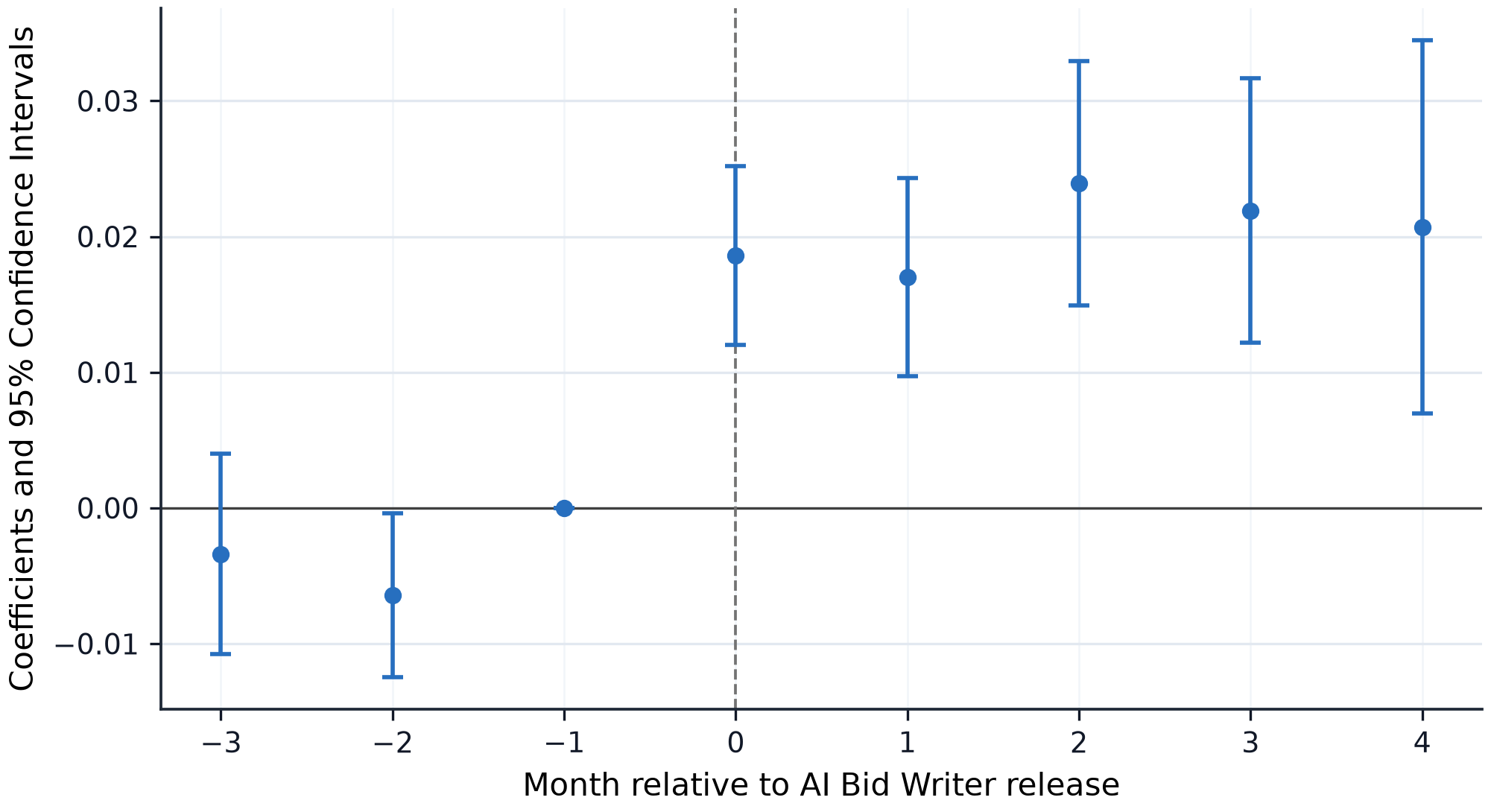}

    \begin{floatnotes}
Cover letter tailoring is measured as text similarity -- the cosine similarity between the TF-IDF vector of a cover letter and that of a job description. The above plots the coefficients and associated confidence intervals of the interactions between treatment and month indicators. Month 0 begins on April 19, 2023, the release of the tool, whereas Month -1 serves as the reference month. Standard errors are clustered at the bidder level.
\end{floatnotes}    
\end{figurebody}
\end{figure}

Our analysis so far has focused on the average effect of AI on cover letter tailoring; however, AI could shift the distribution of the cover letter tailoring in meaningful ways. In Figure \ref{fig:distribution_tfidf}, we plot the distribution of cover letter tailoring for workers who had access to the tool before and after the tool's rollout. We note that the post-rollout sample includes both bids that used the tool and bids that did not, despite having access. Notably, we observe a significant decline in the percentage of cover letters exhibiting extremely low cover letter tailoring. This makes intuitive sense as now the cost of tailoring is drastically reduced.

\begin{figure}[!htbp]
\centering
\begin{figurebody}{0.82\textwidth}
\caption{Distribution of Cover Letter Tailoring Before and After AI Rollout}
\label{fig:distribution_tfidf}

  \centering 
  \includegraphics[width=\linewidth]{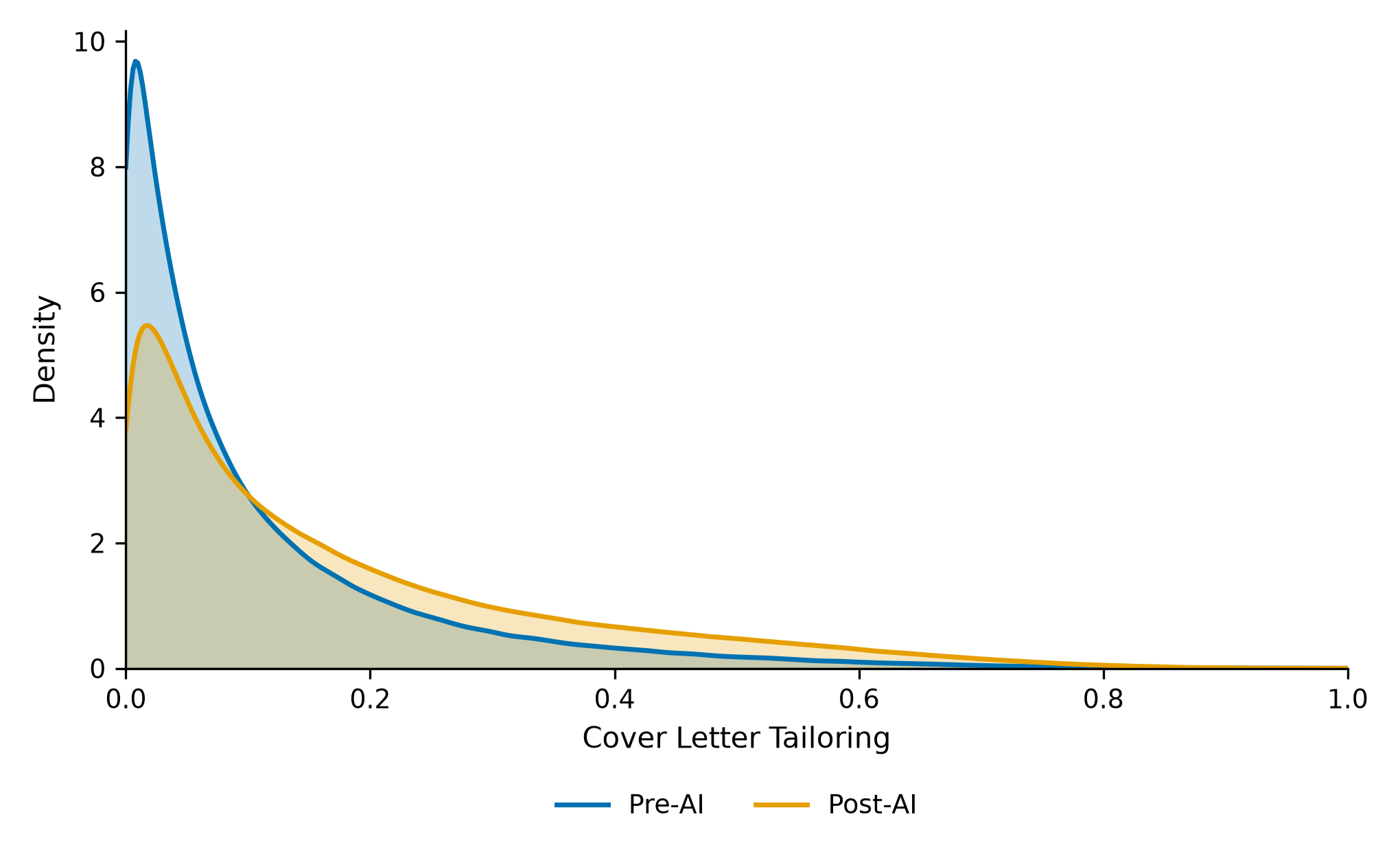}

    \begin{floatnotes}
Cover letter tailoring is measured by text similarity: the cosine similarity between the TF-IDF vector of a cover letter and that of a job description. The sample is limited to bids by workers who had access to the tool. KDE smoothing is applied.
\end{floatnotes}    
\end{figurebody}
\end{figure}

\subsection{Does AI Complement or Substitute for Pre-Existing Writing Ability?}
The previous section found that the AI Bid Writer had a significant and positive impact on cover letter tailoring. However, the impact of generative AI assistance may vary across workers, for instance, depending on their prior ability to craft relevant and tailored cover letters. On the one hand, workers who were already skilled at customizing cover letters may be better positioned to leverage the tool’s capabilities. On the other hand, workers who previously wrote more generic cover letters may benefit more, since they have greater scope for improvement. Our unique data setting allows us to provide direct evidence on the nature of this heterogeneity.

\begin{table}[!htbp]
\centering\tableformat
\caption{Heterogeneous Effect of AI on Cover Letter Tailoring}
\label{tab:tfidf_heterogeneity}
\begin{threeparttable}
\begin{tabular*}{0.90\textwidth}{@{\extracolsep{\fill}}lcc}
\toprule
 & \multicolumn{2}{c}{Cover Letter Tailoring} \\
\cmidrule(lr){2-3}
 & (1) & (2) \\
\midrule
$\text{Post}_t \times \text{Access}_i$ & 0.0255*** & 0.0209*** \\
 & (0.0032) & (0.0032) \\
$\text{Post}_t \times \text{Access}_i \times \text{Pre-AI Ability}_i$ & -0.0112** & -0.0112** \\
 & (0.0045) & (0.0045) \\
\midrule
Bidder FE & \checkmark & \checkmark \\
Month FE & \checkmark & \checkmark \\
$\text{PostGPT}_t \times \text{Access}_{i}$ &  & \checkmark \\
Outcome Mean & 0.1508 & 0.1508 \\
Outcome Std. Dev. & 0.1639 & 0.1639 \\
\midrule
\textit{N} & 2,439,745 & 2,439,745 \\
\bottomrule
\end{tabular*}
\begin{floatnotes}
The table reports regression results with text similarity (TF-IDF) as the outcome. Heterogeneous effects are based on the worker's standardized average TF-IDF in the period before the AI tool was introduced. Column (2) additionally controls for a PostGPT $\times$ Access term. Wage bids, reputation score rank percentile, and log job-level bid count are controlled for. Standard errors are reported in parentheses and clustered at the bidder level. *** p $<$ 0.01, ** p $<$ 0.05, * p $<$ 0.1.
\end{floatnotes}
\end{threeparttable}
\end{table}

To explore whether the impact of the AI writing tool differs based on a worker's pre-AI writing ability, we augment our baseline ITT specification by interacting access to the tool with a worker's pre-AI writing ability $q_i$ as shown below: 
\[
y_{ijt} = \delta_i + \gamma_t + \beta_1 \text{Post AI}_t\times\text{Access}_i + \beta_2 \text{Post AI}_t\times\text{Access}_i \times q_i + x_{ijt}+\epsilon_{ijt}.
\]
The pre-AI writing ability $q_i$ is measured as worker $i$'s average cover letter tailoring before the introduction of the AI writing tool. We standardize $q_i$ to have mean 0 and standard deviation 1. We report the results in Table \ref{tab:tfidf_heterogeneity}. 

We find that the AI tool \emph{substitutes} for workers' pre-existing writing ability. Workers who wrote better cover letters experienced smaller improvements in cover letter tailoring after gaining access to the AI tool compared to those who previously struggled to do so. In the preferred specification with the GPT control, a one-standard-deviation increase in pre-AI writing ability reduces the access effect by 0.0112, or 53.6\% of the average access effect. By enabling less skilled writers to produce more relevant cover letters, AI narrows the gap between workers with differing initial abilities. This suggests that the tool may dilute the informational value of cover letters, a hypothesis we formally test in Section \ref{sec: ge}.

\subsection{The Impact of AI on Callbacks}

Having established that AI improves the tailoring of workers' cover letters, we now use the same specifications to study the impact of AI on a key labor market outcome, the probability of receiving a callback. In our setting, as in most labor markets, employers often reach out to a subset of the bidders for a private conversation, which we refer to as a callback, before making the final hiring decision. We observe which applications and associated cover letters resulted in a callback and use that as the outcome variable. 

We use the same specifications as those in subsection \ref{section: DiD_tailoring} but replace the outcome with an indicator for whether the worker's bid received a callback from the employer. Here, the stable unit treatment value assumption (SUTVA) is potentially violated. The AI tool not only affects workers directly, by altering the cover letters of those who use it and, consequently, their market outcomes, but also creates market-level spillovers. These can include competition spillovers as well as changes in the salience of cover letter tailoring that we study in the following section. The estimand $\beta$ therefore is the sum of the individual effect of gaining access to the AI tool holding the market at the post-AI conditions and the difference in the market-level spillovers between treatment and control workers. Given that treatment and control workers write similarly tailored letters and win auctions at similar rates in the pre-period, it is plausible that the market spillovers are similar for them and our DiD estimate is close to the own-access effect holding market conditions fixed. We provide details of the decomposition in Appendix Section \ref{appendix_sec: decompose}.

\begin{table}[!htbp]
\centering\tableformat
\caption{Effect of AI on Callbacks}
\label{tab:chatted_regression}
\begin{threeparttable}
\begin{tabular*}{0.90\textwidth}{@{\extracolsep{\fill}}lcccc}
\toprule
 & \multicolumn{4}{c}{Callback} \\
 & (1) ITT & (2) LATE & (3) ITT & (4) LATE \\
\midrule
AI & 0.0078*** & 0.0615*** & 0.0038** & 0.0322** \\
   & (0.0016) & (0.0133) & (0.0019) & (0.0159) \\
\midrule
Bidder FE & \checkmark & \checkmark & \checkmark & \checkmark \\
Month FE & \checkmark & \checkmark & \checkmark & \checkmark \\
$\text{PostGPT}_t \times \text{Access}_{i}$ &  &  & \checkmark & \checkmark \\
Outcome Mean & 0.0702 & 0.0702 & 0.0702 & 0.0702 \\
Outcome Std. Dev. & 0.2554 & 0.2554 & 0.2554 & 0.2554 \\
\midrule
\textit{N} & 2,439,745 & 2,439,745 & 2,439,745 & 2,439,745 \\
\bottomrule
\end{tabular*}
\begin{floatnotes}
A callback is defined as being chatted with by the employer after submitting a bid. Columns (1) and (3) report the ITT (the coefficient of PostAI $\times$ Access in the two-way fixed effect specification), while columns (2) and (4) report the LATE, where we instrument for the usage of AI using PostAI $\times$ Access. Columns (1) and (2) omit the PostGPT $\times$ Access term; columns (3) and (4) include it. Wage bids, reputation score rank percentile, and log bid count are controlled for. Standard errors are clustered at the bidder level. *** p $<$ 0.01, ** p $<$ 0.05, * p $<$ 0.1.
\end{floatnotes}
\end{threeparttable}
\end{table}


Table \ref{tab:chatted_regression} shows positive effects on callbacks. In column (3), access to the AI Bid Writer increases the probability of receiving a callback by 0.38 percentage points, or 5.20\% relative to the $7.37\%$ baseline callback rate among treated workers. When we instrument for usage (LATE), the estimated effect rises to 3.22 percentage points, or 43.78\% of the baseline rate. These estimates are economically meaningful and statistically significant at the 5\% level, though not at the 1\% level. However, Appendix Figure \ref{fig:honestdid_callbacks} shows that the ITT estimate is not robust to plausible violations of the parallel trends assumption. The breakdown value à la \citet{rambachan2023more} is negative for the main specification and less than 1 for the specification using weekly observations and dropping the period before the GPT-4 rollout. We therefore view the callback results as suggestive evidence that access to the AI Bid Writer increased workers' likelihood of receiving a callback, rather than as a robust finding.

\begin{figure}[!htbp]
\centering
\begin{figurebody}{0.82\textwidth}
\caption{Dynamic Effects of AI on Callbacks (ITT)}
\label{fig:chatted_did}

  \centering 
  \includegraphics[width=\linewidth]{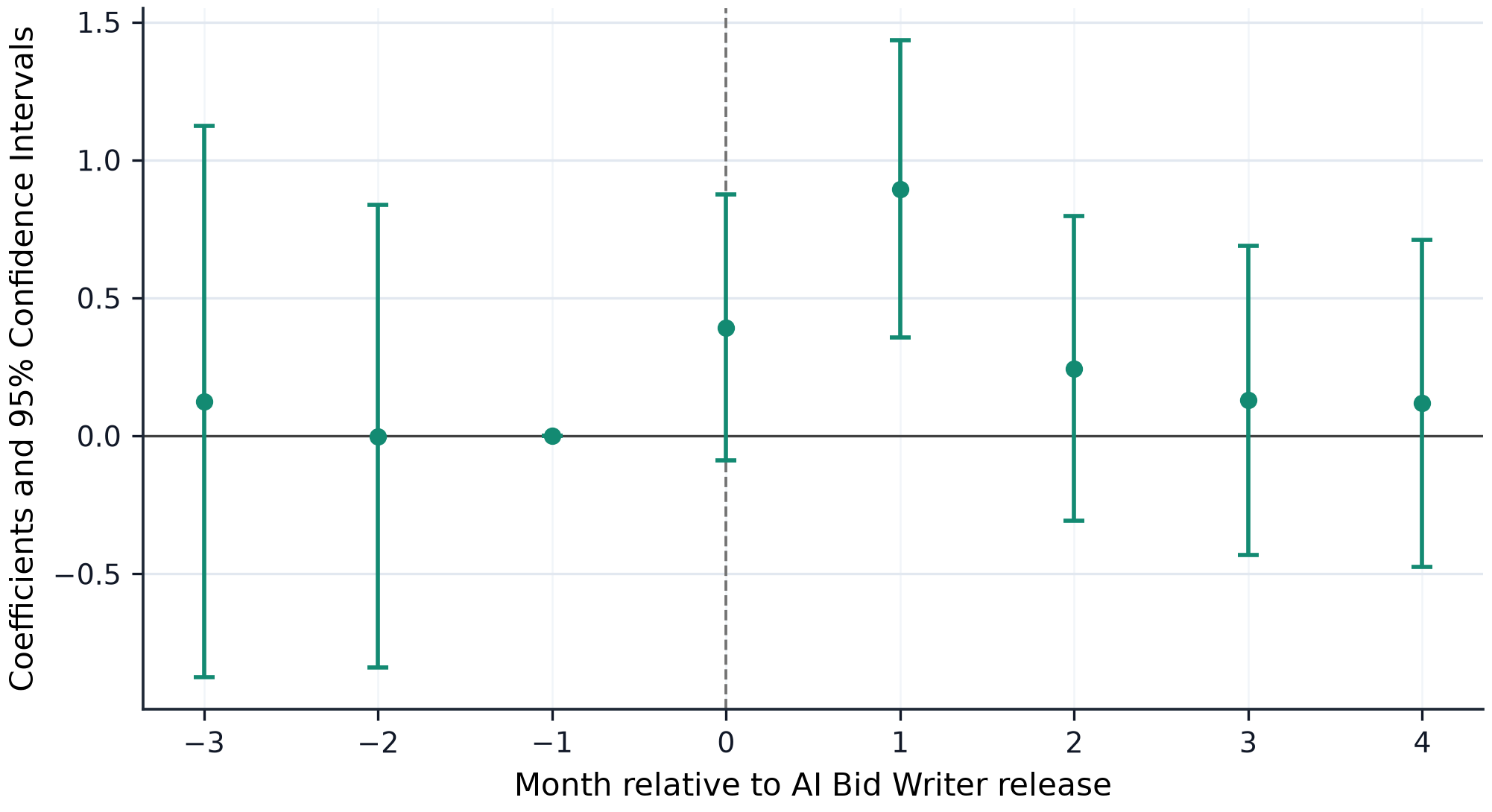}

    \begin{floatnotes}
A callback is defined as being chatted with by the employer after submitting a bid. The above plots the coefficients and associated confidence intervals of the interactions between treatment and month indicators. Month 0 begins on April 19, 2023, the release of the tool, whereas Month -1 serves as the reference month. Standard errors are clustered at the bidder level.
\end{floatnotes}    
\end{figurebody}
\end{figure}

Figure \ref{fig:chatted_did} plots the coefficients and confidence intervals from the dynamic ITT specification. The point estimates are positive in the first two months after the tool's introduction, but the confidence intervals are wide and include zero in most post-rollout months. By contrast, the dynamic effects of AI on cover letter tailoring remain consistently positive. 



Our results therefore provide strong evidence that AI improves cover letter tailoring and suggestive evidence that this improvement may translate into higher callback rates for workers who had access to the tool. Tools that help workers quickly customize their cover letters may enable them to better signal interest in the project and convey alignment with the role, thereby increasing their chances of passing the initial screening stage and being invited to interview. However, there are at least two reasons why greater cover letter tailoring may not translate into meaningful labor market outcomes such as callbacks. First, if employers know that workers have access to tools that automate cover-letter customization, they may recognize that features that previously signaled worker interest or fit have become less informative and therefore place less weight on them. Second, if employers can identify which cover letters were generated using AI, they may directly discount AI-generated letters. Either channel could weaken or eliminate the effect of improved tailoring on callbacks. We investigate these two channels directly in the following sections.

\FloatBarrier

\section{\label{sec: ge} The Market-Level Impact of Generative AI}

An important advantage of our setting is that the AI tool was adopted in a substantial share of applications, allowing us to analyze its market-level impact. In this section, we examine how the tool's introduction affected the informativeness of three bid characteristics observed by employers: cover letter tailoring, workers' reputation, and normalized wage bids. We then investigate in more depth the change in the signaling value of cover letter tailoring. The section concludes with a discussion on welfare. 

\subsection{Changes in the Predictive Power of Bid Characteristics}

The individual-level results show that AI has its largest effects on workers who previously wrote more generic cover letters. By narrowing differences in cover letter tailoring across workers, the tool may weaken the relationship between written applications and workers' underlying ability, fit, or effort. If employers recognize this change in the information environment, they may discount cover letter tailoring and place greater weight on alternative signals. On Freelancer.com, one salient alternative is workers' prior work history and the reviews they have received from previous employers. This information is displayed both in workers' applications and on their profiles and also enters the platform's visual ranking of bids. We use the platform-calculated score summarizing this information as a measure of workers' experience and reputation. Wage bids provide another observable characteristic that may enter employers' screening decisions directly or serve as a signal of otherwise unobserved worker quality.

To study these margins, we estimate the following event study specification:
\[
y_{ijt} = \delta_i + \gamma_t + \beta X_{ijt}' + \theta\text{PostAI}_t \times X_{ijt}' +\alpha \ln n_{j}+ \epsilon_{ijt},
\]
where $y_{ijt}$ is an indicator for whether the bid by worker $i$ for job $j$ in month $t$ receives an interview or an offer, $\delta_i$ are bidder fixed effects, $\gamma_t$ are month fixed effects; $X_{ijt}$ contains cover letter tailoring, reputation score, and normalized wage bid; $\ln n_{j}$ is the log number of bids received by job $j$. We additionally control for log project-level bid count. The coefficients in $\theta$ capture how the relationships between each bid characteristic and employer callback and hiring decisions changed after the introduction of the AI Bid Writer.

\begin{table}[!htbp]
\centering\tableformat

\caption{Predictive Power of Bid Characteristics Before and After AI}
\label{tab:eventstudy_combined}
\begin{threeparttable}
\begin{tabular*}{0.90\textwidth}{@{\extracolsep{\fill}}lcc}
\toprule
& Callback & Offer \\
\midrule
$\text{Cover Letter Tailoring}_{ijt}$
& 0.036603*** & 0.003382*** \\
& (0.002383) & (0.000658) \\
$\text{Post}_t \times \text{Cover Letter Tailoring}_{ijt}$
& -0.018942*** & -0.003401*** \\
& (0.002605) & (0.000677) \\

\midrule

$\text{Reputation Score}_{ijt}$
& 0.058666*** & -0.004349*** \\
& (0.007388) & (0.001666) \\
$\text{Post}_t \times \text{Reputation Score}_{ijt}$
& 0.013035*** & -0.000116 \\
& (0.004076) & (0.000758) \\

\midrule

$\text{Normalized Wage Bid}_{ijt}$
& -0.001754*** & -0.000212*** \\
& (0.000101) & (0.000027) \\
$\text{Post}_t \times \text{Normalized Wage Bid}_{ijt}$
& -0.000167* & 0.000011 \\
& (0.000098) & (0.000021) \\

\midrule
Bidder FE & \checkmark & \checkmark  \\
Month FE & \checkmark & \checkmark  \\
\midrule
\textit{N} & 5{,}651{,}635 & 5{,}651{,}635 \\
\bottomrule
\end{tabular*}
\begin{floatnotes}
The table reports bid-level regressions that include all three bid characteristics and their interactions with the post-AI indicator. Cover-letter tailoring is the cosine similarity between the TF-IDF vectors of the cover letter and job post. Reputation score is a platform-calculated score based on the volume and strength of a bidder's past reviews. Normalized wage is the bid amount divided by the project's minimum budget. Bidder and month fixed effects are included. Log job-level bid count is controlled for. Standard errors are reported in parentheses and clustered at the bidder level.  *** p $<$ 0.01, ** p $<$ 0.05, * p $<$ 0.1.
\end{floatnotes}

\end{threeparttable}
\end{table}

Table \ref{tab:eventstudy_combined} shows that, before the introduction of the AI tool, cover letter tailoring strongly predicts both stages of employer selection. A one-unit increase in tailoring is associated with a 3.66 percentage point higher probability of receiving a callback and a 0.34 percentage point higher probability of receiving an offer. After the rollout, these relationships decline sharply: the callback coefficient falls by 1.89 percentage points, or 51.75\% of the pre-period relationship, and the offer coefficient falls by 0.34 percentage points—more than fully offsetting the pre-period relationship and leaving no detectable association between tailoring and offers in the post-AI period.

Examining the coefficients on reputation score, we find that the interview coefficient becomes 22.22\% stronger after the rollout, consistent with employers placing significantly greater weight on workers' past work experience and reputation when deciding whom to interview. The corresponding change for offers is economically small (2.67\%) and statistically insignificant, suggesting that the shift toward reputation score operates primarily at the interview stage rather than at the final hiring stage.\footnote{We estimate a negative relationship between workers' reputation score and offer at baseline. This could be due to several reasons, including not instrumenting for endogenous outcomes such as wage bids \citep{cui2025bidding} and workers becoming picky with accepting job awards.} 

When considering employers' screening responses to wages, the predicted effect is less clear \textit{ex ante}: although wages enter directly into employers’ payoffs, they may also serve as a signal of workers’ unobservable characteristics \citep{delacroix2013pricing,chen2024wages}. After the AI tool's rollout, the correlation between the probability of being
interviewed and workers’ bids remains negative and does not change significantly. These results suggest that, at least in the short term, employers shifted
attention toward reputation scores without significantly changing their sensitivity to wages.

\subsection{The Decline in the Signaling Value of Cover Letters: Mechanisms}

Most notably, the preceding analysis documents a substantial weakening of the relationship between cover letter tailoring and employer decisions. We now examine this decline in greater depth to understand how AI altered the value of cover letters as a signal in employer screening.

\par\addvspace{0.5\baselineskip}
\noindent \textbf{AI vs.\ Non-AI Bids.} A natural question following the prior analysis is whether the weakening relationship between cover letter tailoring and employer decisions is concentrated among AI-assisted bids or whether it also extends to non-assisted bids, the latter of which would be consistent with employers discounting cover letters more broadly. We therefore estimate a pooled specification that allows the tailoring slope to differ across AI and non-AI assisted bids.

Figure \ref{fig:split_estudy_ai_nonai} plots the coefficients across pre-AI bids, post-release bids which were not assisted by AI, and post-release AI-assisted bids. For callbacks, the coefficient on tailoring declines from 0.0377 before the rollout to 0.0242 among post-AI non-assisted bids and 0.0114 among post-AI assisted bids, corresponding to declines of 35.8\% and 69.8\%, respectively. For hiring, the gradient falls by 77.9\% among non-assisted bids and 122.3\% among assisted bids. The decline is larger for AI-assisted bids, suggesting that employers may have been able to identify them. This could help explain why workers using the AI tool did not experience callback gains. Importantly, the coefficient also declines substantially and significantly among bids that did not use the platform AI tool. This supports a market-level discounting interpretation: employers reduced the weight placed on cover letter tailoring across the applicant pool following the rollout. 

\begin{figure}[!htbp]
\centering
\begin{figurebody}{\textwidth}
\caption{Predictive Power of Cover Letter Tailoring by AI tool Usage}
\label{fig:split_estudy_ai_nonai}

  \centering 
  \includegraphics[width=\linewidth]{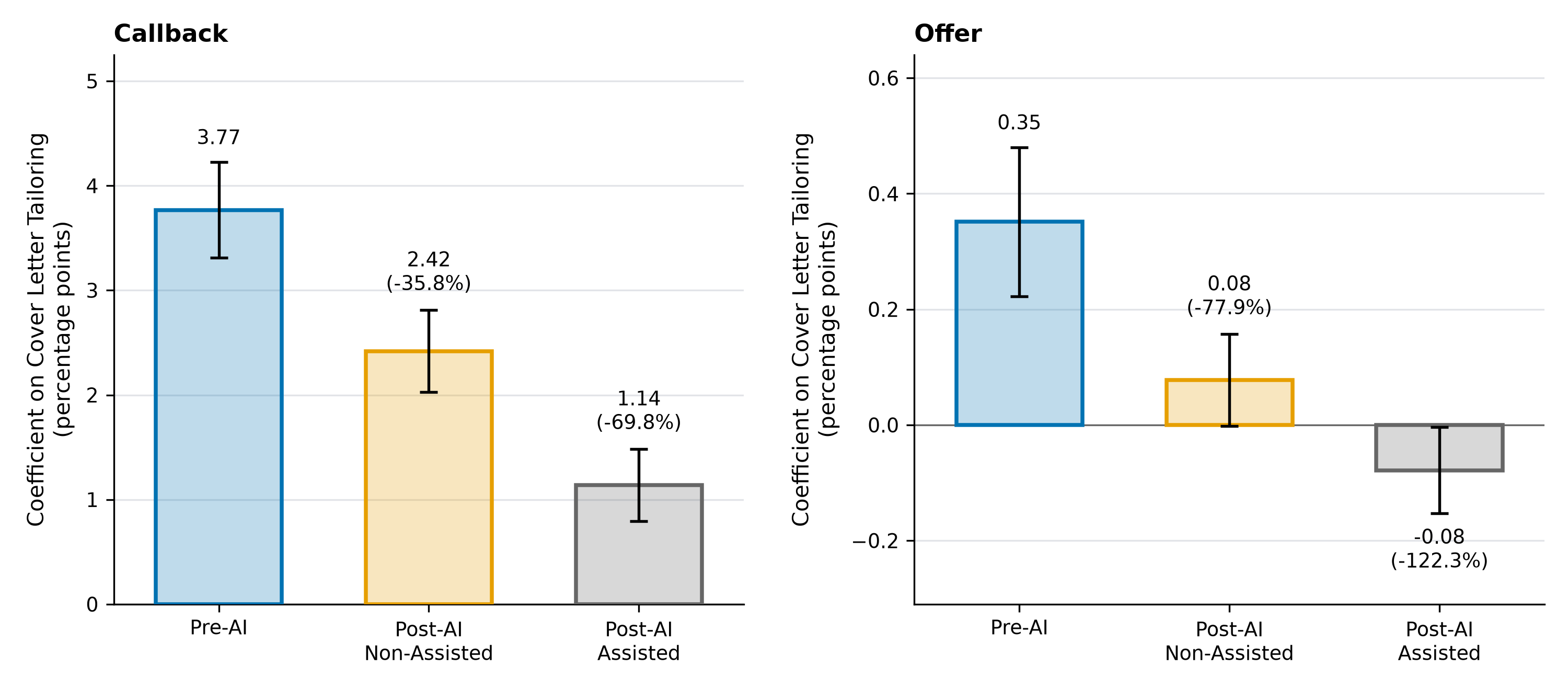}

    \begin{floatnotes}
The figure plots the coefficient on cover letter tailoring from a bid-level regression of callback and offer on tailoring, allowing the tailoring slope to differ across three groups: bids submitted before the AI tool's release (Pre-AI), post-release bids that did not use the tool (Post-AI Non-Assisted), and post-release bids that did use the tool (Post-AI Assisted). Normalized wage and reputation score, each interacted with the post-AI indicator, and log(bid count) are included as controls. Bidder and month fixed effects are included, and standard errors are clustered at the bidder level.
\end{floatnotes}    
\end{figurebody}
\end{figure}

\par\addvspace{0.5\baselineskip}
 \noindent \textbf{Informativeness.} The preceding event-study estimates show that the conditional relationship between cover letter tailoring and employer decisions flattens after the tool rollout. However, a flatter coefficient captures only the marginal association per unit of measured tailoring; if the rollout compresses the distribution of tailoring, the estimated relationship may weaken even though the remaining variation in tailoring still helps employers distinguish among bids. As a complementary exercise, we therefore compare the incremental explanatory power of measured tailoring for employer decisions before and after the release of the tool. To this end, we estimate the bid-level callback and offer regressions on bid characteristics before and after the rollout and obtain the partial $R^2$ of tailoring.

 The results largely mirror the findings from the event study. For callbacks, the partial $R^2$ of cover letter tailoring falls from 0.0277 percent before the rollout to 0.0101 percent afterward, an absolute decline of 0.0177 percentage points and a relative decline of 63.7 percent. For offers, it falls from 0.0018 percent to almost zero, a relative decline of 99.6 percent. While the baseline partial $R^2$ levels are small, particularly for the low-base-rate offer outcome, their large proportional declines are consistent with the event-study estimates and suggest that the tailoring of cover letters contributes substantially less incremental predictive content after the rollout. In Appendix Figure~\ref{fig:partialr2_ai_nonai}, we additionally examine this decline separately for AI and non-AI assisted bids with similar findings to the previous analysis. While the partial $R^2$ associated with AI-assisted bids falls more, a large degree of informativeness is still lost among bids which did not use the tool.

\par\addvspace{0.5\baselineskip}
 \noindent \textbf{Quantile Effects} Another alternative explanation of the observed event-study patterns is mechanical. If employers treat cover letter tailoring as a nonlinear or threshold signal, AI may move many bids into regions of the cover letter tailoring distribution where additional tailoring has lower marginal value. This can cause the static coefficient on tailoring to fall even if employers' threshold rule remains the same. To assess this possibility, we estimate the tailoring premium across pre-AI tailoring score deciles and estimate callback and offer premiums separately before and after AI:
 \[
Y_{ijt}
=
\alpha_i
+
\gamma_t
+
\sum_{k=2}^{10}
\beta_k \mathbf{1}\{\text{Tailoring Decile}_{ijt}=k\}
+
\theta X_{ijt}
+
\varepsilon_{ijt}.
\]

Figure \ref{fig:quantiles} plots the callback and offer premiums relative to the lowest non-zero tailoring score decile and shows that employers generally do not assess tailoring with a threshold rule when making callback or hiring decisions: callback and offer premiums rise across the tailoring distribution. After the release of the AI tool, however, the premium is lower across all deciles both for receiving a callback and an offer. Thus, the decline in the predictiveness of cover letter tailoring cannot solely be explained by a mechanical distributional shift in tailoring. Instead, the evidence indicates that high tailoring itself became less informative after AI, consistent with employers discounting cover letter tailoring once it became easier to produce.

\begin{figure}[!htbp]
\centering
\begin{figurebody}{\textwidth}
\caption{Cover Letter Tailoring Premium Before and After AI Rollout, by Decile}
\label{fig:quantiles}

  \centering 
  \includegraphics[width=\linewidth]{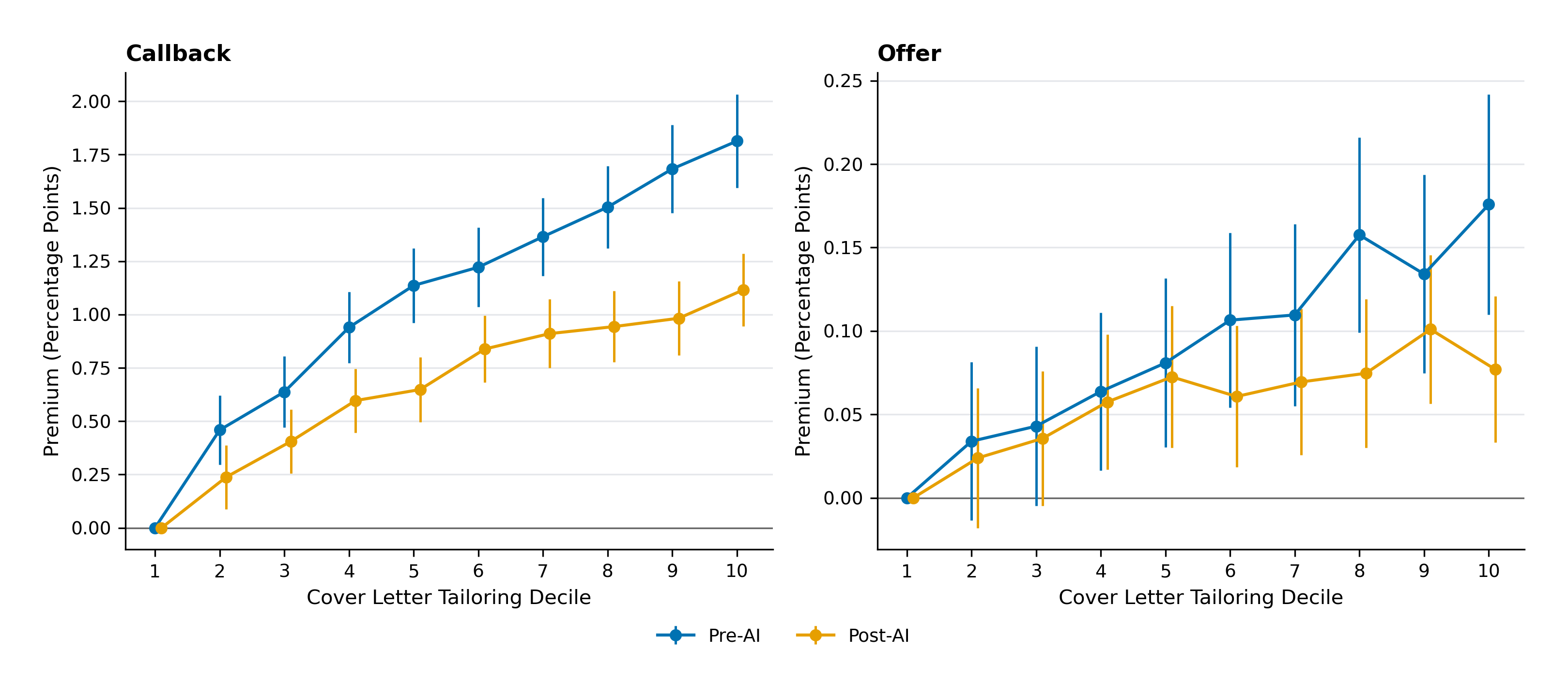}

    \begin{floatnotes}
The figure plots the premium in callback and offer probability associated with each cover letter tailoring decile, estimated separately for bids submitted before and after the AI tool's release. Cover letter tailoring is measured as the cosine similarity between the TF-IDF vectors of a cover letter and a job description. Bids with exactly zero tailoring are excluded. Deciles are constructed from the pre-AI bids and applied to both periods. Decile 1 is the reference group. Each regression controls for normalized wage bid, reputation score, and log job-level bid count, with bidder and month fixed effects. Standard errors are clustered at the bidder level.
\end{floatnotes}    
\end{figurebody}
\end{figure}

\par\addvspace{0.5\baselineskip}
\noindent \textbf{Timing and Persistence} To evaluate the timing and persistence of the decline in the signaling power of cover letters, we estimate the following dynamic event-study specification:
\begin{equation}
\label{eq: dynamic_event_study}
\begin{aligned}
y_{ijt} ={}& \delta_i + \gamma_t
  + \sum_{k \neq -1} \beta_k \, \mathbf{1}\{t = k\}\text{Cover Letter Tailoring}_{ijt} \\
 &+ \text{PostGPT}_t \times \text{Cover Letter Tailoring}_{ijt}
  + x_{ijt} + \epsilon_{ijt},
\end{aligned}   
\end{equation}
where $t$ is month, $y_{ijt}$ is whether the cover letter by worker $i$ for job $j$ in month $t$ received a callback or offer, $\text{Cover Letter Tailoring}_{ijt}$ measures cover letter tailoring, and $x_{ijt}$ includes the usual controls. 
Figure \ref{fig:dynamic_tfidf} presents the coefficients on month indicators times cover letter tailoring, which are significantly negative in several months following the rollout of AI Bid Writer, suggesting that having greater cover letter tailoring is less advantageous post AI-rollout. The decline occurred in the first month after the AI tool was introduced. It was not confined to a short-lived disruption at launch: highly tailored letters remained less predictive of both interviews and offers. 


\begin{figure}[!htbp]
\centering
\begin{figurebody}{0.80\textwidth}
\caption{Post AI Roll-out, Cover Letters Persistently Less Predictive of Interview and Offer}
\label{fig:dynamic_tfidf}

  \centering 

  \includegraphics[width=0.94\linewidth]{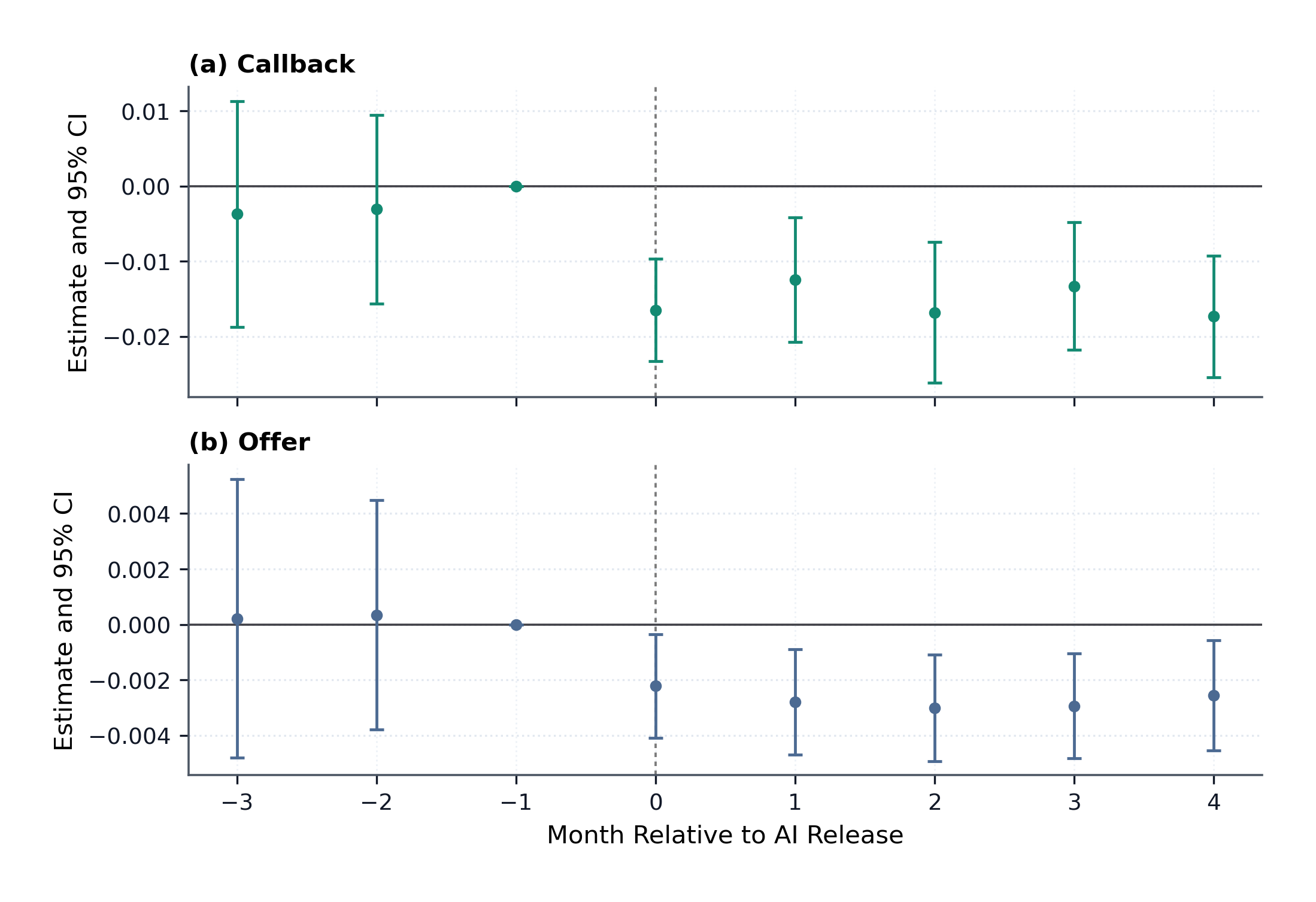}

    \begin{floatnotes}
The figure plots interaction-term event-study coefficients on Month $\times$ Cover Letter Tailoring. Panel A uses callback as the outcome. Panel B uses offer as the outcome. Month $-1$ is omitted. Normalized wage bid, reputation score, log job-level bid count, bidder fixed effects, and month fixed effects are controlled for. Standard errors are clustered at the bidder level.
\end{floatnotes}    
\end{figurebody}
\end{figure}

\par\addvspace{0.5\baselineskip}
\noindent \textbf{The Role of Employer Experience} The preceding analysis shows that employers persistently placed less weight on cover letter tailoring after the release of the AI tool. A natural question is whether this adjustment reflects Bayesian learning by employers. According to this hypothesis, through hiring workers and observing their on-the-job performance, employers learn that the relationship between cover letters and workers' ability has changed, and then down-weight the signal from cover letters. Our setting, however, provides limited scope for assessing this mechanism. Freelancer.com is primarily a task-based marketplace where employers primarily hire for one-time projects. Few employers make repeated hiring decisions for which feedback from earlier matches might inform how they screen future applicants.

While there is hence little scope for learning through repeated hiring, experience may nevertheless shape screening responses to the change in the informational environment. Employers who had awarded projects prior to the rollout had reviewed at least one applicant pool, selected a worker, and observed the resulting match. Thus, these employers entered the post-AI period with greater familiarity with the previous application environment. Once AI expanded the supply of polished cover letters, prior hiring experience may have made the screening change more salient.

To examine whether the decline in the predictive value of cover letter tailoring differs by employer experience, we estimate the dynamic event-study in Equation \ref{eq: dynamic_event_study} separately for employers
with and without prior hiring experience, using callbacks as the outcome. We classify an employer as
experienced if they awarded at least one project prior to the release of the tool. We additionally compare the change in the tailoring coefficient between experienced and inexperienced employers using a cross-model Wald test. Standard errors are clustered at the bidder level, accounting for bidder overlap in the two samples. 

Panel A of Figure \ref{fig:employer_split} plots the event-study coefficients of cover letter tailoring for experienced and inexperienced employers separately. Panel B plots the experienced-minus-inexperienced difference. 

The tailoring coefficient declines in both groups, with a larger decline among experienced employers. Averaged over months \(0\)--\(4\), the change is \(-0.0102\) for inexperienced employers and \(-0.0251\) for experienced employers. Relative to their respective reference-month slopes of \(0.0358\) and \(0.0493\), these changes represent declines of \(28.5\%\) and \(51.0\%\). Thus, the predictive relationship between tailoring and callbacks weakens by just over one-quarter among inexperienced employers and by roughly one-half among experienced employers. In most of the post-AI months, the difference in the response between the experienced and inexperienced employers is significant at the 95\% confidence level.

Notably, employers with no observed pre-period award also discounted tailoring. Although they lacked a benchmark from an observed pre-AI platform match, they could possibly infer that tailoring had become less diagnostic from public awareness of generative AI or from the unusually large number of polished letters in their current applicant pool. Employers with prior platform experience had access to these same sources of inference, but also possessed personal hiring experience under the pre-AI regime, which may explain their stronger response.

\begin{figure}[!htbp]
\centering
\begin{figurebody}{\textwidth}
\caption{Predictiveness of Cover Letters for Callbacks by Employer Experience}
\label{fig:employer_split}

  \centering

  \includegraphics[width=\linewidth]{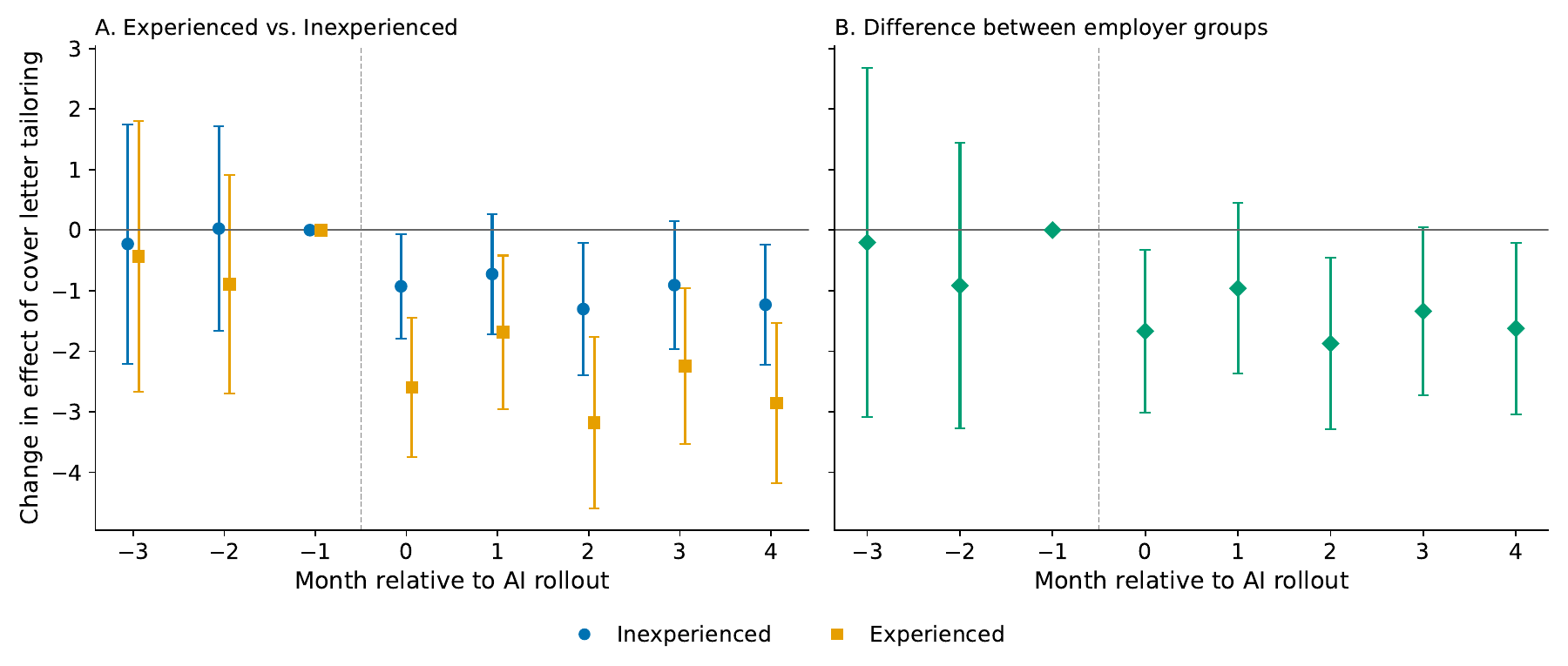}

    \begin{floatnotes}
The figure plots interaction-term event-study coefficients on Month $\times$ Cover Letter Tailoring with interview as the outcome, separately by whether the employer awarded at least one project before the AI rollout. Month $-1$ is omitted. Normalized wage, reputation score rank percentile, bidder fixed effects, and month fixed effects are controlled for. Standard errors are clustered at the bidder level.
\end{floatnotes}    
\end{figurebody}
\end{figure}

\subsection{Welfare Analysis}

Evaluating the full welfare consequences of AI requires a model-based analysis that is beyond the scope of this paper. Nevertheless, we examine several indicators that would enter such an analysis and assess how they changed in the months following the automation of cover letters for a subset of workers.

Interview, hiring, and job completion rates are important components of total welfare. Figure \ref{fig:job_means} plots the share of jobs that resulted in at least one interview, an award, or successful completion before and after the introduction of the AI tool. We find no discernible changes across any of these outcomes: interview rates remain stable, as do the probabilities of job award and eventual completion.

\begin{figure}[!htbp]
\centering
\begin{figurebody}{0.82\textwidth}
\caption{Share of Jobs with Interviews, Awards, and Completions}
\label{fig:job_means}

  \centering 
  \includegraphics[width=\linewidth]{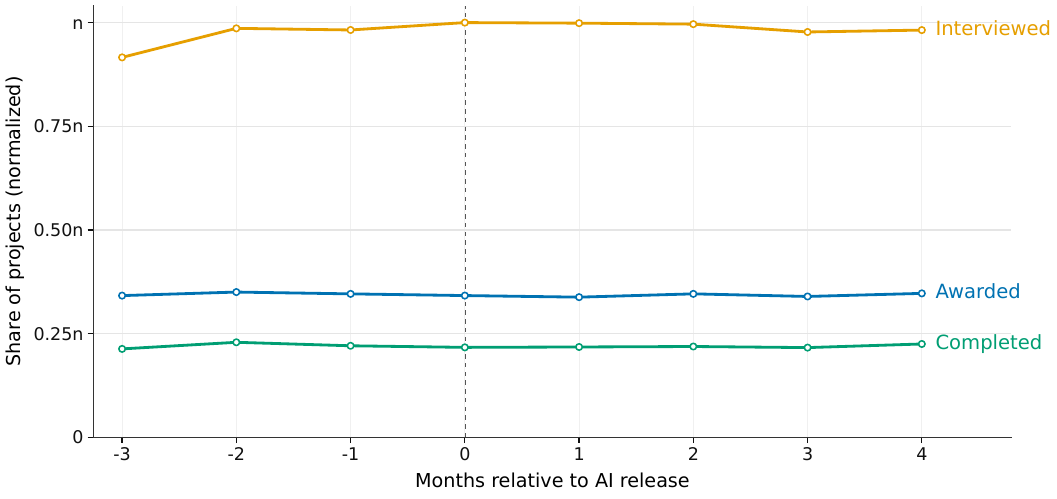}

    \begin{floatnotes}
The above figure plots the percentage of jobs such that (i) the employer engaged in at least one interview; (ii) the employer awarded the job; and (iii) the employee successfully completed the job. A total of 106,714 jobs are represented. Due to data confidentiality, the y-axis is normalized. The maximum value is shown as ‘n’, and other ticks represent fractions of n.
\end{floatnotes}    
\end{figurebody}
\end{figure}

The absence of movement in hiring rates may appear surprising given the substantial change of informational content documented above. Nonetheless, this pattern is consistent with the interpretation that while AI disrupted one channel of signaling---cover letters---employers successfully shifted toward others, particularly work histories, which partially substituted for the diminished informativeness of written applications. In this sense, the market may have adjusted to maintain equilibrium matching rates, despite changes in the informational environment.

Another margin through which the tool could have affected the market equilibrium is wage bids and realized wages. Our earlier estimates show that employers' sensitivity to wage bids did not change significantly after the tool was introduced. To examine whether workers changed how much they bid, we track bids over time. As Figure \ref{fig:wage_panel} shows, the average wage bid, the average winning bid, and the within-project standard deviation of bids all remain stable after the rollout. 

\begin{figure}[!t]

\centering
\begin{figurebody}{\textwidth}
\caption{Bid Outcomes Over Time}
\label{fig:wage_panel}\label{fig: bid_patterns}

\centering

\begin{subfigure}{0.48\textwidth}
    \centering
    \wagepanel{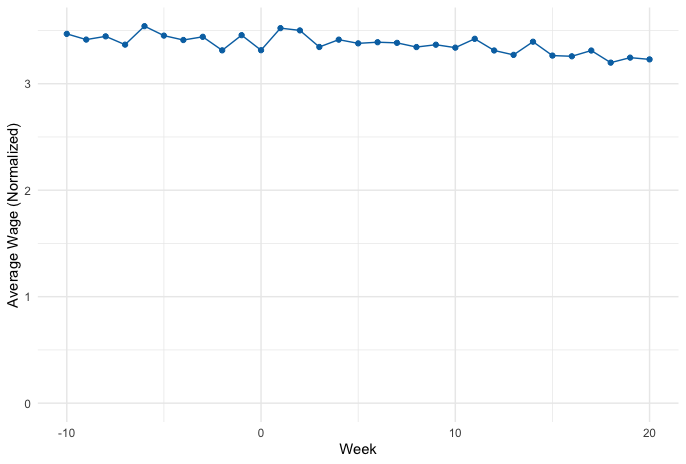}{38bp 42bp 7bp 7bp}
      {0.0437/$-10$,0.3479/0,0.6505/10,0.9532/20}
      {0.0458/0,0.3036/1,0.5590/2,0.8145/3}
      {Average wage (normalized)}{Week}{}
    \caption{Average Wage Bid}
    \label{fig:wage_norm}
\end{subfigure}
\hfill
\begin{subfigure}{0.48\textwidth}
    \centering
    \wagepanel{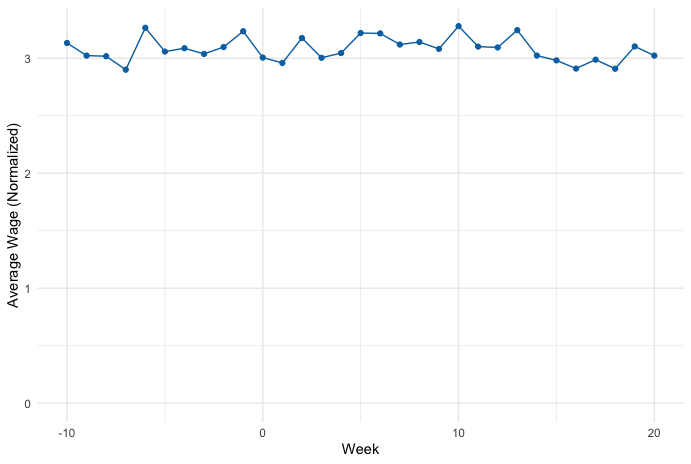}{38bp 42bp 8bp 7bp}
      {0.0434/$-10$,0.3488/0,0.6512/10,0.9550/20}
      {0.0458/0,0.3229/1,0.6000/2,0.8771/3}
      {Average wage (normalized)}{Week}{}
    \caption{Average Winning Bid}
    \label{fig:wage_winner}
\end{subfigure}

\par\vspace{8pt}

\begin{subfigure}{0.48\textwidth}
    \centering
    \wagepanel{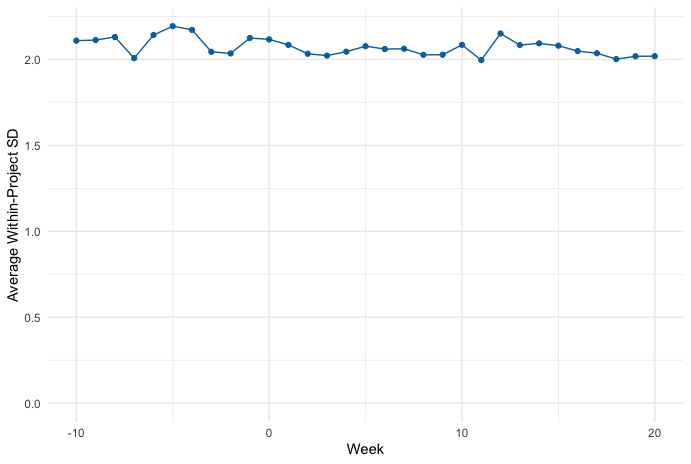}{47bp 42bp 8bp 7bp}
      {0.0456/$-10$,0.3491/0,0.6525/10,0.9560/20}
      {0.0458/0.0,0.2530/0.5,0.4602/1.0,0.6675/1.5,0.8747/2.0}
      {Average within-project SD}{Week}{}
    \caption{Average Wage Bid Standard Deviation by Project}
    \label{fig:wage_sd}
\end{subfigure}
\hfill
\begin{subfigure}{0.48\textwidth}
    \centering
    \wagepanel{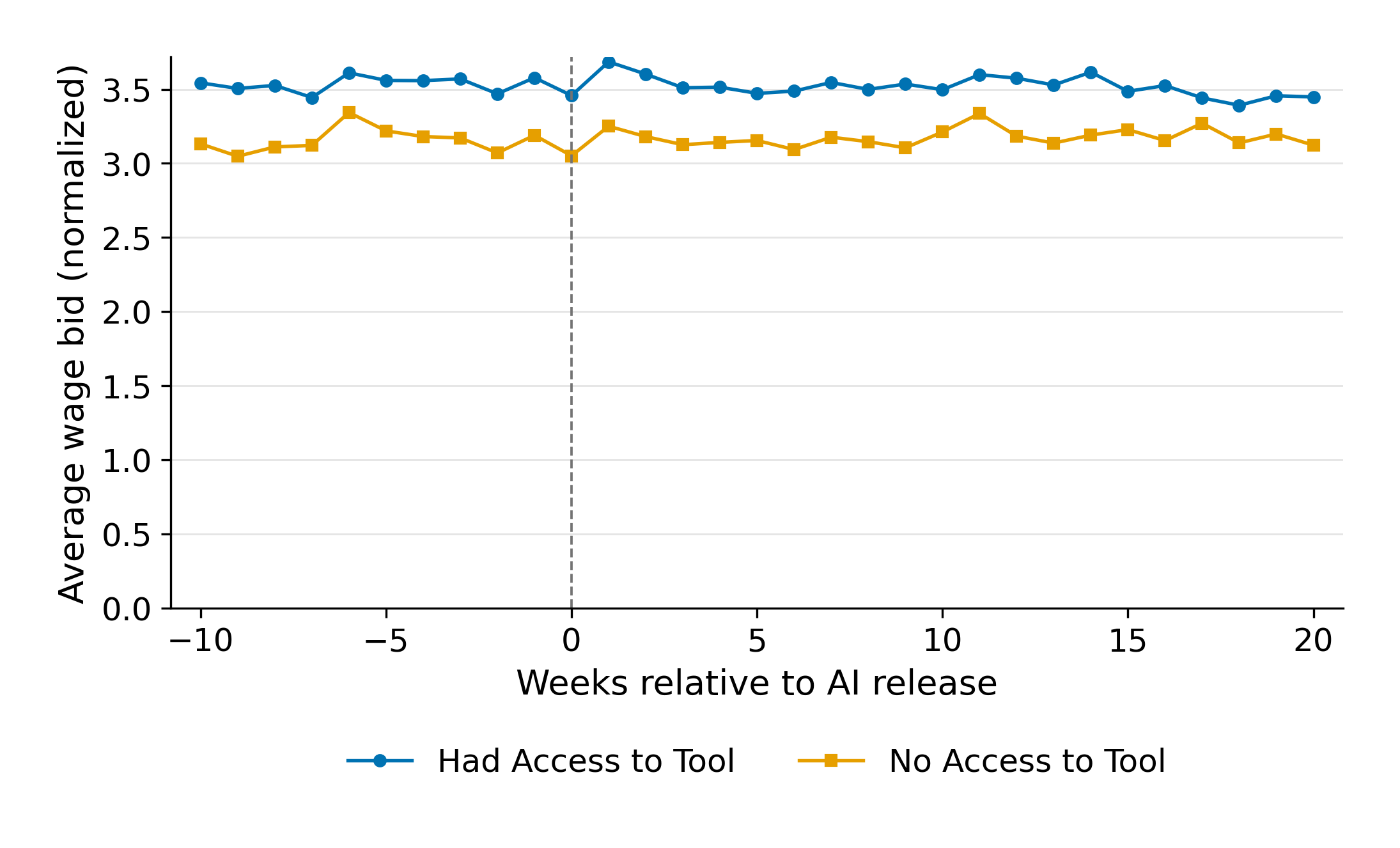}{63.36bp 89.28bp 20.64bp 21.12bp}
      {0.0272/$-10$,0.3424/0,0.6576/10,0.9728/20}
      {0.0023/0.0,0.1364/0.5,0.2705/1.0,0.4046/1.5,0.5387/2.0,0.6728/2.5,0.8069/3.0,0.9410/3.5}
      {Average wage bid (normalized)}{Weeks relative to AI release}
      {\textcolor{econblue}{\rule{9pt}{0.8pt}}\ Had access to tool\quad
       \textcolor{econgold}{\rule{9pt}{0.8pt}}\ No access to tool}
    \caption{Average Bid by Treatment Group}
    \label{fig:wage_by_treatment_weekly}
\end{subfigure}

\begin{floatnotes}
Panel (a) reports weekly average bids normalized by each project's minimum budget.
Panel (b) reports weekly average winning bids normalized by each project's minimum budget.
Panel (c) reports the average within-project standard deviation of bids. Panel (d) reports weekly average bids normalized by each project's minimum budget separately for workers who gained access to the AI tool when it was rolled out and those who did not. The AI tool was rolled out in Week 0. 
\end{floatnotes}

\end{figurebody}
\end{figure}

Our results suggest that, with respect to bids, neither workers nor firms substantially changed their strategies or beliefs. In the short run, the signaling role of cover letters appears to have shifted toward past work experience, while the role of bids remained largely unchanged. From a welfare perspective, this suggests that the way surplus is divided between firms and workers did not change meaningfully in the short run; rather, the change appears to concern only how surplus is allocated across workers.

With respect to differences across workers, Figure \ref{fig:wage_by_treatment_weekly} plots the evolution of wage bids for workers who gained access to the AI tool and those who did not separately. The figure suggests no meaningful change in bidding strategies following the introduction of the tool in either group. Taken together with the previous results, this implies that the relative welfare gains for treated workers operated primarily through higher callback probabilities, rather than through higher earnings conditional on winning.


We view the welfare interpretation in this paper as preliminary and emphasize two important caveats. First, our analysis focuses on the immediate aftermath of the AI tool’s release, when cover letters remained one of several signals available to employers. Over time, as AI tools become more sophisticated and begin to erode additional dimensions of signaling, employers may find themselves with fewer credible indicators on which to rely. Second, our setting involves short-term freelance jobs where past reviews are readily available and highly salient. In other labor markets, such as those for first-time job seekers or college admissions, alternative signals may be weaker or absent, leaving fewer substitutes when written applications lose their informational value.


\FloatBarrier

\section{\label{sec: human_ai} Human-AI Interaction}

In the final section of the paper, we leverage unique data to understand worker interactions with the platform-native tool. We examine the extent to which workers edit the AI-generated drafts, the factors that shape their editing effort, and whether that effort influences hiring outcomes.

\subsection{Workers' Time Spent Editing AI-Generated Output}
While the AI tool aligns the content of the letter with the job requirements, workers are free to further customize the letters. We first document whether workers edit the output generated by the AI tool or simply submit it unchanged. To measure editing, we use the elapsed time between when a worker clicked the AI cover-letter generation tool and when the application was submitted, interpreting this duration as editing time. Since this measure is only available for bids that used the tool, our analysis is restricted to AI-assisted applications. 

Figure~\ref{fig:hist_min_5plus} plots the distribution of editing times. The majority of AI-assisted cover letters appear to be submitted with little or no modification: over 75\% were finalized within one minute of clicking the tool. By contrast, a small share of workers invested more effort in customization, with roughly 5\% of bids submitted at least five minutes after the AI draft was generated.

\begin{figure}[!htbp]
\centering
\begin{figurebody}{0.82\textwidth}
\caption{Distribution of Time Spent Editing AI Output (Minutes)}
\label{fig:hist_min_5plus}

  \centering

    \includegraphics[width=\linewidth]{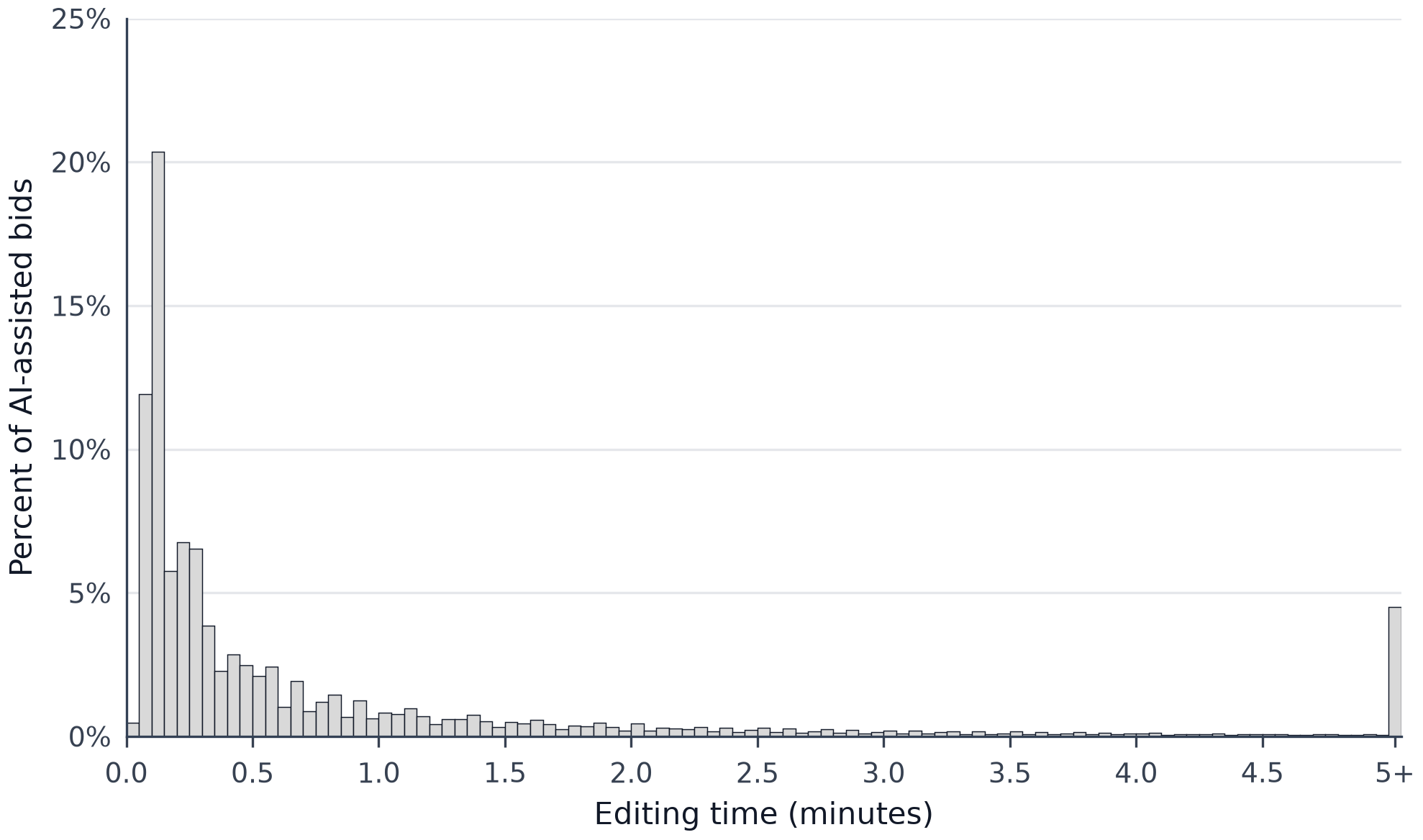}
  
      \begin{floatnotes}
Bid-level histogram. Each observation represents the time, in minutes, between activating the tool and submitting the bid. Observations of 5 minutes or more are grouped in the final bin labeled ``5+''.
\end{floatnotes}
\end{figurebody}
\end{figure}
We next examine which types of workers invest effort in editing the tool's output. Specifically, we estimate a worker-level regression of the average time between activating the tool and submitting a bid on the worker's average pre-period cover letter tailoring. The goal is to assess how workers with different levels of cover-letter writing skill interact with AI: conditional on using the tool, do stronger writers simply accept its output, or do they devote more time to tailoring it?

\begin{table}[!htb]
\centering\tableformat
\caption{Determinants of Time Spent Editing AI-Generated Output}
\label{tab:mean_time_tfidf}
\begin{threeparttable}
\begin{tabular*}{0.75\textwidth}{@{\extracolsep{\fill}}lc}
\toprule
 & Average Editing Time (Minutes) \\
\midrule
Average pre--AI Cover Letter Tailoring & 3.136*** \\
                        & (0.608)  \\
Constant                & 1.209*** \\
                        & (0.084)  \\
\midrule
$N$                     & 2,366    \\
\bottomrule
\end{tabular*}
\begin{floatnotes}[\justifying]
The unit of observation is a worker. The dependent variable is the worker's average editing time, measured in minutes between activating the AI tool and submitting a bid. The explanatory variable is the worker's mean cover letter tailoring, ranging from 0 to 1, during the period before the tool's introduction. The estimation sample consists of workers treated at rollout who submitted at least one post-rollout bid using the tool. Bid-level editing times below zero or greater than or equal to 60 minutes are excluded before calculating worker averages. Standard errors are heteroskedasticity-robust. *** $p<0.01$, ** $p<0.05$, * $p<0.1$.
\end{floatnotes}
\end{threeparttable}
\end{table}

Table~\ref{tab:mean_time_tfidf} reports the results. The coefficient is positive and highly significant, indicating that higher pre-AI writing ability is associated with greater editing effort. To gauge the magnitude, moving from the 10th to the 90th percentile of pre-period cover letter tailoring among treated workers corresponds to an increase of approximately $0.70$ minutes in average editing time. Relative to the bid-level distribution, this increase represents approximately 59\% of the mean editing time ($1.18$ minutes), or $0.20$ standard deviations (the standard deviation is $3.50$ minutes). These findings highlight heterogeneity in how workers engage with AI: more skilled writers invest additional effort in customizing their letters. Appendix Figure~\ref{fig:hist_mean_5plus_two} compares the distributions of editing times for workers with average tailoring scores above and below the pre-AI median. Approximately 68\% of workers below the median in pre-AI writing ability have an average editing time of less than one minute, compared with 51\% of workers above the median. Editing effort among above-median workers is instead spread over longer durations.

\subsection{Correlation between Editing Time and Hiring Outcome}
Having documented patterns in how workers interact with the tool, we now examine whether workers' time spent editing influences hiring outcomes, namely the likelihood of receiving a callback and the likelihood of receiving a job offer. To this end, we regress indicators for whether a bid received a callback or an offer on the time that worker spent editing the AI-generated cover letters. We control for worker fixed effects, as well as normalized wages and reputation score rank percentiles.

\begin{table}[!htb]\centering\tableformat
\caption{Correlation Between Time Spent Editing AI-Generated Output and Hiring Outcomes}
\label{tab:time_min}
\begin{threeparttable}
\begin{tabular*}{0.90\textwidth}{@{\extracolsep{\fill}}lcc}
\toprule
 & Offer & Callback \\
\midrule
\addlinespace
Time spent editing (minutes) &
0.000908*** &
0.001991*** \\
& (0.000139) & (0.000227) \\
\midrule
$N$ & 615,150 & 615,150 \\
Worker FE & \checkmark & \checkmark \\
Outcome mean & 0.0044 & 0.0616 \\
\bottomrule
\end{tabular*}

\begin{floatnotes}
All regressions include a constant and controls for normalized wage and rank percentile. Offer is a bid-level indicator equal to 1 if the bid won the job and 0 otherwise. Callback is a bid-level indicator equal to 1 if the worker who submitted the bid was interviewed for the task and 0 otherwise. Time spent editing is measured in minutes from the moment the worker activated the tool to the moment the bid was submitted. The sample consists of all post-rollout bids that used the tool and had an editing time of less than 60 minutes. Standard errors, clustered at the worker level, are reported in parentheses. *** $p<0.01$, ** $p<0.05$, * $p<0.1$.
\end{floatnotes}
\end{threeparttable}
\end{table}

The results are reported in Table \ref{tab:time_min}. The Offer column reports results using job offer as the dependent variable. The coefficient on editing time is positive and highly significant: spending more minutes tailoring a cover letter after using the tool is associated with a higher probability of winning. A one–standard deviation increase in editing time ($3.22$ minutes) is associated with a $0.29$ percentage-point higher probability of winning, or about 66\% of the average win rate among AI-assisted bids (0.44\%). In the Callback column, the coefficient on editing time is likewise positive and highly significant. A one–standard deviation increase corresponds to a $0.64$ percentage-point higher probability of chatting, roughly 10\% of the average chat rate among AI-assisted bids (6.16\%).

These results underscore the importance of how workers interact with AI. Section~\ref{sec:pe} showed that access to the AI tool may have increased the likelihood of receiving a callback in the short run. Here we add another layer of nuance: spending more time refining one’s cover letter after adopting the tool is associated with higher probabilities of both receiving a callback and ultimately winning the job. However, these results should not be interpreted as causal: the decision to spend more time editing a cover letter is unlikely to be random and may correlate with unobserved bid characteristics that employers find appealing. Further work is needed to more fully understand human–AI interaction in the context of job applications.

\section{\label{sec: conclusion} Conclusion}

This study provides empirical evidence on the labor market consequences of AI-assisted job applications. While AI tools allow freelancers to produce more polished and tailored applications with less effort, our findings suggest that they fundamentally reshape how employers interpret cover letters. The widespread adoption of AI-assisted writing diminishes the informational value of cover letters, weakening their role as a hiring signal. In response, employers place greater weight on alternative signals, such as past work experience on the platform, that are less susceptible to AI. We find no evidence that this shift in signaling has affected overall matching rates and wages. Finally, we show that although most workers submit AI-generated drafts with minimal revision, within a worker, more time spent editing is associated with higher hiring success.

While online labor platforms such as Freelancer provide a unique opportunity to study the growing use of AI in job applications, it is important to consider how our findings may extend to other labor markets. The basic hiring process on Freelancer resembles that in many other settings: employers receive applications from multiple workers, use application materials to infer workers' fit and interest, and invite a subset of applicants to interview. This similarity suggests that the mechanism we document—AI reducing the informational content of application materials as employers learn to discount signals that have become easier to produce—may extend beyond online labor platforms. At the same time, several features of our setting are important for interpreting the external validity of our results. First, Freelancer primarily facilitates short-term projects and provides employers with unusually rich and easily accessible information about workers' past performance, including reviews and experience on similar projects. To the extent that such credible alternative signals are less available in other labor markets, a decline in the informativeness of cover letters could have more adverse informational consequences than those we document here. Second, although workers on free membership plans can apply for jobs, Freelancer limits the number of applications workers can submit, with these limits varying by membership plan. These constraints may help explain why we do not observe a large increase in application volume following the introduction of the AI tool. In labor markets where applying is less constrained, AI could generate a much larger increase in application volume, potentially overwhelming employers and producing different effects on screening and hiring. Finally, we observe the effects of the tool over a six-month period, during which employers already begin to down-weight the cover-letter signal. Other equilibrium responses by workers and employers could take longer to materialize. As AI technologies continue to advance and diffuse, continued research will be important for understanding their longer-run implications for labor market signaling and screening.

\newpage
\begingroup
\renewcommand*{\bibfont}{\normalfont\small\singlespacing}
\appto\bibsetup{\raggedright\microtypesetup{protrusion=false,expansion=false}}
\setlength{\bibitemsep}{5pt}
\printbibliography
\endgroup


\newpage

\clearpage
\ifdefined\processdelayedfloats
\processdelayedfloats  
\fi

\appendix
\section*{Appendix}

\section{Tables and Figures}

\setcounter{table}{0}
\renewcommand{\thetable}{A\arabic{table}}
\setcounter{figure}{0}
\renewcommand{\thefigure}{A\arabic{figure}}

\begin{apptable}
\begingroup
  \captionsetup{type=table,width=\textwidth,aboveskip=0pt,belowskip=8pt}
  \singlespacing                               
  \setlength{\tabcolsep}{4pt}                  
  \renewcommand{\arraystretch}{1.03}           
  \newcommand{\tightpar}{\setlength{\parskip}{0pt}}
  \newcolumntype{L}{>{\raggedright\arraybackslash\tightpar}p{10.8cm}}
  \newcolumntype{C}{>{\centering\arraybackslash}m{1.8cm}}
  \newcolumntype{D}{>{\centering\arraybackslash}m{2.0cm}}

  \captionof{table}{Examples of Cover Letter Tailoring}
  \label{tab:similarity_scores}
  \vspace{6pt}

  \begin{minipage}{0.95\textwidth}
    \fontsize{9}{10.5}\selectfont
    \setlength{\parskip}{0pt}   
    \textbf{Brief:} I need to add a horizontal scrolling section to a WordPress website.
    User is scrolling vertical (up \& down) while content is moving left and right.
    Solution: Plain HTML/CSS/JS is ok. Even better if we can use class/ID on the WPBakery Page Builder row to activate.
    Has to work on the demo site.

    I have some examples that can be copied (CSS and JS methods). Demo site: [LINK REDACTED]
  \end{minipage}

  \vspace{0.4em}

  \fontsize{9}{10.5}\selectfont
  \begin{threeparttable}
\begin{tabular}{L C D}
    \toprule
    \textbf{Proposal Text} & \textbf{Cover Letter Tailoring} & \textbf{AI Bid Writer} \\
    \midrule

    Hello! I understand that you need to add a horizontal scrolling section to your WordPress website, but the user needs to be able to scroll up and down while content is moving left and right. I also check the attached url with reference and I can make this perfectly.

    My name is Sandeep and I am a Full Stack Developer with over 7+ years of experience in the industry. I specialize in larger web applications and software development, as well as larger website builds.

    My solution is to use class/ID on the WPBakery Page Builder row to activate the section. This method works on the demo site and doesn't require any modifications to the site's code.

    If you are interested in my services please feel free to take a look at my portfolio or contact me directly for more information. Thank you for your consideration!
    & 0.541 & Yes \\
    \midrule

    Hello
    I am a senior Web and Mobile app developer and I have worked on many jobs similar to this one. I understand that you need to add a horizontal scrolling section to your wordpress website, but instead of just moving from left to right, the content needs to move from left to right while the user is still able to scroll up and down through the section. My skillset allows me to tackle jobs like this and bring the user an exciting experience.

    I am confident that my expertise can help bring this job to success. If given the chance, I would love to work on this site with you! Please don't hesitate to reach out if you have any questions or would like more information about me or my skillset.

    Thank you for your consideration!
    & 0.354 & Yes \\
    \midrule

    Hi,
    I can help u to add Vertical Scrolling function to ur WP Website
    I checked the demo site and found that It's easy task for my team because we had 10 + years experience in it
    Plz chat now bro !
    & 0.339 & No \\
    \midrule

    "Hi Joshua M.!

    Because of my considerable experience in website development, I believe I would be a wonderful fit for this
    position. I have examined every detail, and everything appears doable, but if you could start the chat, I have a few
    questions.

    Here's what I can bring to your job: 1) Top-notch experience 2) Knowledge of best practices for data security 3)
    24/7 communication: I will keep you updated with job's status and consistent updates 4) Someone who cares
    about helping you succeed and bring more value to your organization

    I am looking forward to hearing more about your exciting job!"
    & 0.001 & No \\
    \bottomrule
  \end{tabular}

  \begin{floatnotes}[\justifying]
A selection of four written proposals on an example job. Cover letter tailoring is computed as the cosine similarity between the TF-IDF vectors of the job description and the proposal. AI Bid Writer indicates whether the bid was written with the help of the platform-native AI tool.
\end{floatnotes}
\end{threeparttable}
\endgroup
\end{apptable}

\begin{apptable}
\centering\tableformat
\captionof{table}{Comparison of Observables between Full and DiD Samples}
\label{tab:samplecomp}
\begin{threeparttable}
\begin{tabular*}{0.90\textwidth}{@{\extracolsep{\fill}}lrrrr}
\toprule
& Full & Trimmed & Difference & SMD \\
\midrule
\multicolumn{5}{l}{\textit{Panel A. Bidder-Level}} \\
\hspace{0.75em}Weeks Active & 1.979 & 7.814 & 5.835 & 0.879 \\
\hspace{0.75em}Bids per Active Week & 2.743 & 5.956 & 3.213 & 0.193 \\
\hspace{0.75em}Cover Letter Tailoring & 0.123 & 0.137 & 0.014 & 0.108 \\
\hspace{0.75em}Normalized Wage Bid & 2.915 & 3.081 & 0.166 & 0.096 \\
\hspace{0.75em}Reputation Score & 0.569 & 0.632 & 0.063 & 0.691 \\
\hspace{0.75em}Callback Rates & 0.052 & 0.067 & 0.016 & 0.103 \\
\hspace{0.75em}Award Rates & 0.003 & 0.008 & 0.005 & 0.111 \\
\hspace{0.75em}$N$ & 282{,}789 & 18{,}248 & - & - \\
\midrule
\multicolumn{5}{l}{\textit{Panel B. Bid-Level}} \\
\hspace{0.75em}Cover Letter Tailoring & 0.154 & 0.151 & -0.003 & -0.020 \\
\hspace{0.75em}Normalized Wage Bid & 3.339 & 3.494 & 0.154 & 0.047 \\
\hspace{0.75em}Reputation Score & 0.768 & 0.851 & 0.083 & 0.554 \\
\hspace{0.75em}Callback Rates & 0.064 & 0.070 & 0.006 & 0.024 \\
\hspace{0.75em}Award Rates & 0.005 & 0.007 & 0.001 & 0.016 \\
\hspace{0.75em}$N$ & 5{,}800{,}993 & 2{,}439{,}745 & - & - \\
\bottomrule
\end{tabular*}
\begin{floatnotes}[\justifying]
This table compares observables in the full sample and the trimmed estimation sample used in the individual-impact analysis, which excludes workers who are not active in both the pre- and post-periods and those whose access to the tool changed. Reputation score is a platform-calculated score summarizing a worker's work history at the time of bidding. The bidder-level panel collapses the data to the bidder level before computing summary statistics. The bid-level panel reports statistics at the bid level. Difference is defined as Trimmed Sample minus All Bids. Standardized mean differences (SMDs) are computed using the pooled standard deviation. The statistics are specific to our data sample and are not platform-wide. 
\end{floatnotes}
\end{threeparttable}
\end{apptable}

\begin{apptable}
\centering\tableformat
\captionof{table}{Comparison of Pre-Period Observables by AI Tool Use}
\label{tab:complier_characteristics}
\begin{threeparttable}
\begin{tabular*}{0.90\textwidth}{@{\extracolsep{\fill}}lrrrr}
\toprule
 & Never used & Ever used & Difference & SMD \\
\midrule
\multicolumn{5}{l}{\textit{Panel A. Bidder-Level}} \\
\hspace{0.75em}Weeks Active & 5.792 & 8.052 & 2.260 & 0.512 \\
\hspace{0.75em}Bids per Active Week & 9.296 & 25.369 & 16.073 & 0.422 \\
\hspace{0.75em}Cover Letter Tailoring & 0.116 & 0.125 & 0.008 & 0.086 \\
\hspace{0.75em}Normalized Wage Bid & 3.165 & 3.228 & 0.063 & 0.035 \\
\hspace{0.75em}Reputation Score & 0.755 & 0.782 & 0.027 & 0.184 \\
\hspace{0.75em}Callback Rates & 0.095 & 0.080 & -0.0145 & -0.102 \\
\hspace{0.75em}Award Rates & 0.014 & 0.013 & -0.0004 & -0.007 \\
\hspace{0.75em}$N$ & 1{,}467 & 2{,}386 & - & - \\
\midrule
\multicolumn{5}{l}{\textit{Panel B. Bid-Level}} \\
 & Non-AI bids & AI bids & Difference & SMD \\
\hspace{0.75em}Cover Letter Tailoring & 0.116 & 0.122 & 0.006 & 0.070 \\
\hspace{0.75em}Normalized Wage Bid & 3.550 & 3.509 & -0.040 & -0.033 \\
\hspace{0.75em}Reputation Score & 0.867 & 0.855 & -0.012 & -0.091 \\
\hspace{0.75em}Callback Rates & 0.075 & 0.070 & -0.0051 & -0.091 \\
\hspace{0.75em}Award Rates & 0.007 & 0.007 & 0.0002 & 0.011 \\
\hspace{0.75em}$N$ & 1{,}209{,}401 & 190{,}692 & - & - \\
\bottomrule
\end{tabular*}
\begin{floatnotes}[\justifying]
This table is restricted to workers with access at rollout in the DiD estimation sample. Tool use is measured from April 19 to September 15, 2023. Panel A compares workers who used the tool at least once with those who never used it during that period, weighting each bidder equally. Bidders' characteristics are measured from January 19 to April 18, 2023. Panel B compares post-rollout AI-assisted and unassisted bids, attaching each bidder's pre-period averages to every bid. Weeks active counts pre-period weeks with at least one sample bid; bids per active week divides pre-period bids by weeks active. Reputation score summarizes past work and reviews; wages are normalized by the minimum job budget. Difference is Ever used minus Never used in Panel A and AI bids minus Non-AI bids in Panel B. SMD divides the difference by the square root of the average of the two group sample variances. $N$ counts bidders in Panel A and post-rollout bids in Panel B. The statistics are specific to our data sample and are not platform-wide. 
\end{floatnotes}
\end{threeparttable}
\end{apptable}

\begin{appfigure}
\centering
\begin{figurebody}{\textwidth}
    \centering
  \captionof{figure}{HonestDiD $M$-Sensitivity Plot (Tailoring as Outcome)}
    \label{fig:honestdid_tfidf}

    \begin{subfigure}[t]{0.48\textwidth}
        \centering
        \caption{Monthly Specification}
        \includegraphics[width=\linewidth]{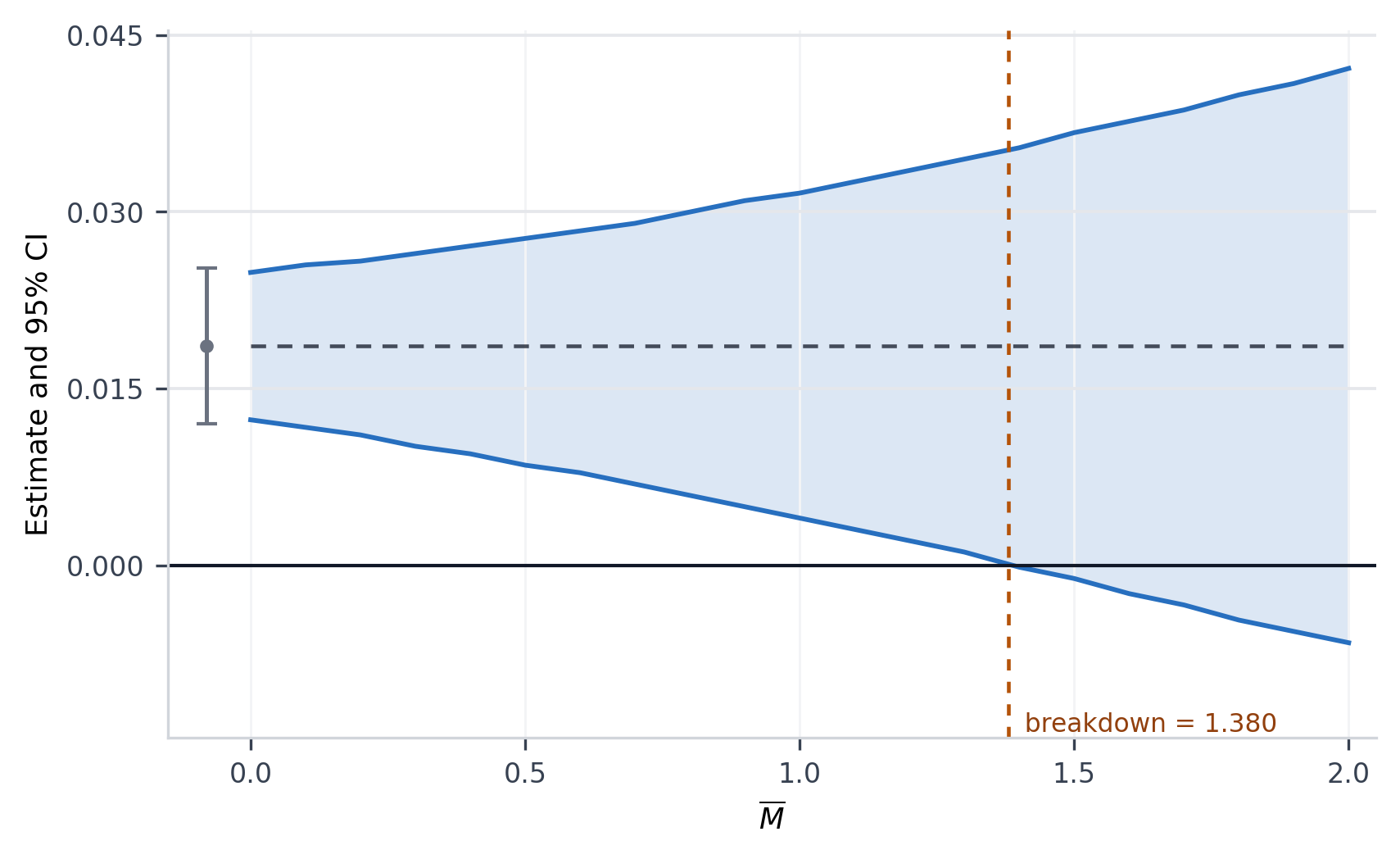}
        \label{fig:honestdid_monthly_tfidf}
    \end{subfigure}
    \hfill
    \begin{subfigure}[t]{0.48\textwidth}
        \centering
        \caption{Weekly Post-GPT Specification}
        \includegraphics[width=\linewidth]{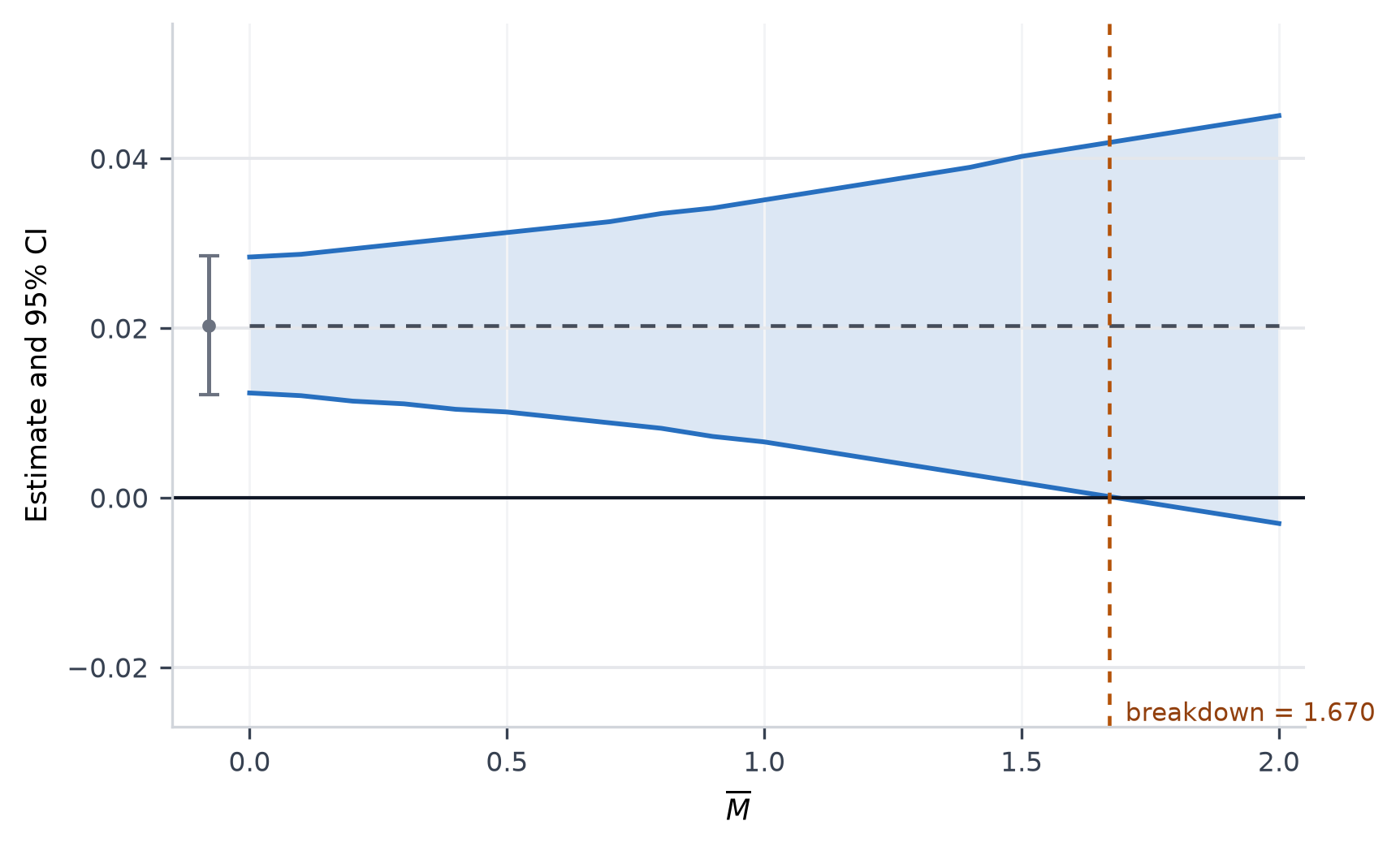}
        \label{fig:honestdid_weekly_tfidf}
    \end{subfigure}

    \begin{floatnotes}
This figure reports Rambachan and Roth (2023) sensitivity intervals for the first monthly/weekly post-AI-release effect on cover letter tailoring. The monthly specification includes bidder fixed effects, month fixed effects, and controls for post-GPT4 $\times$ Access interaction, normalized wage bid, reputation score rank percentile, and log job-level bid count. The weekly specification limits the pre-period to the four weeks between the release of GPT-4 and the release of the platform-native tool. It includes bidder fixed effects, week fixed effects, and controls for normalized wage bid, reputation score rank percentile, and log job-level bid count.
\end{floatnotes}
\end{figurebody}
\end{appfigure}

\begin{appfigure}
\centering
\begin{figurebody}{\textwidth}
    \centering
  \captionof{figure}{HonestDiD $M$-Sensitivity Plot (Callbacks as Outcome)}
    \label{fig:honestdid_callbacks}

    \begin{subfigure}[t]{0.48\textwidth}
        \centering
        \caption{Monthly Specification}
        \includegraphics[width=\linewidth]{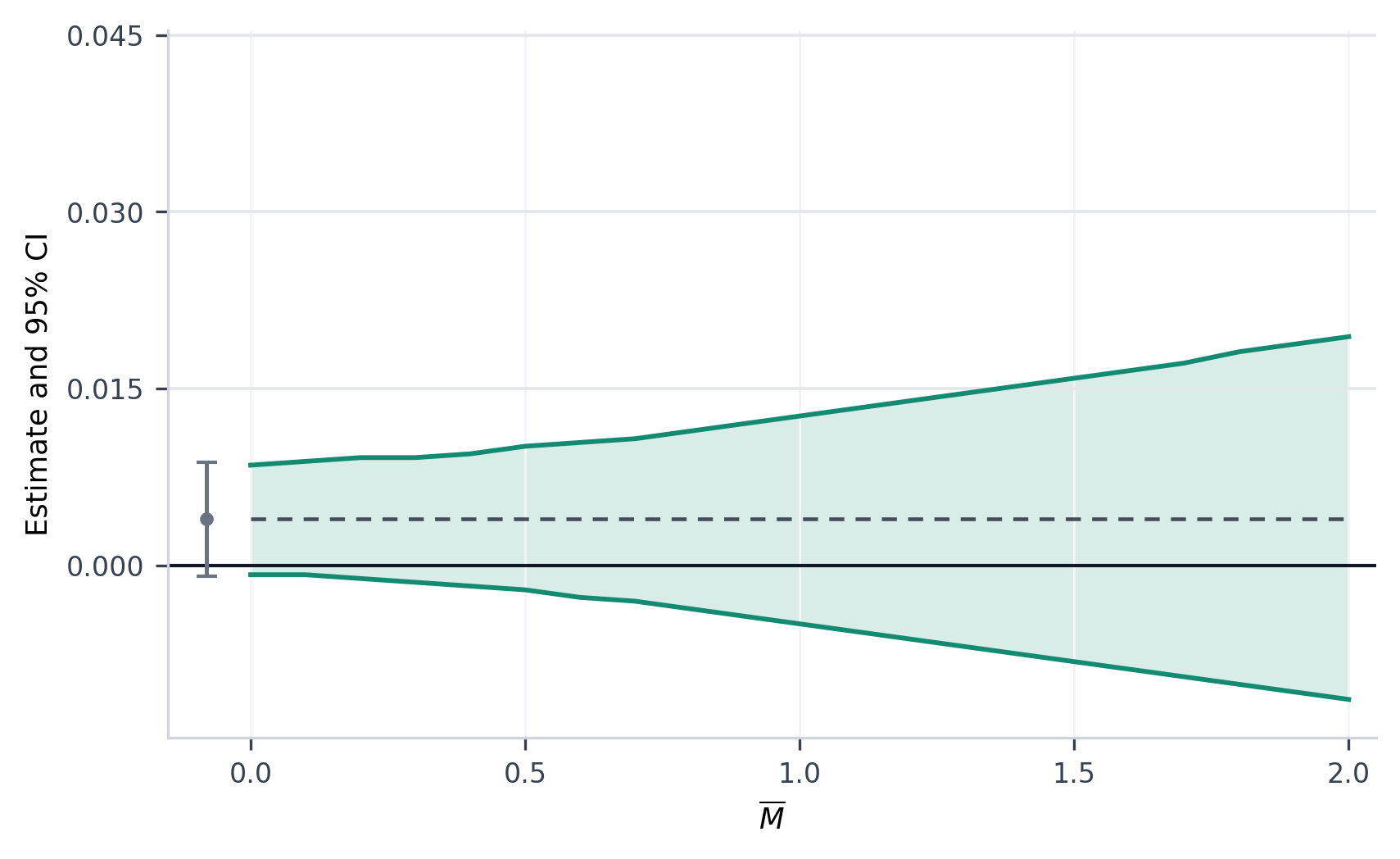}
        \label{fig:honestdid_monthly_callback}
    \end{subfigure}
    \hfill
    \begin{subfigure}[t]{0.48\textwidth}
        \centering
        \caption{Weekly Post-GPT Specification}
        \includegraphics[width=\linewidth]{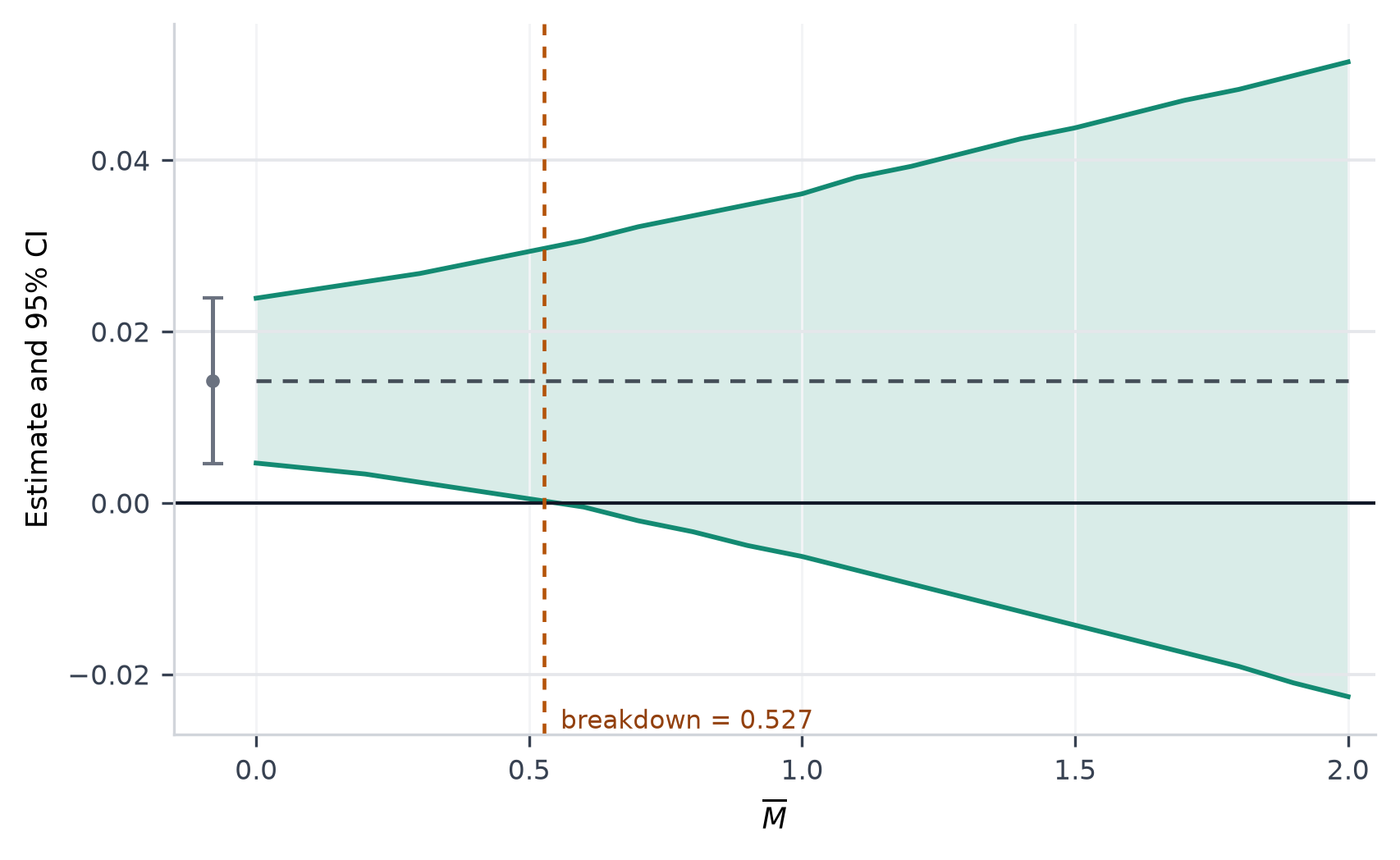}
        \label{fig:honestdid_weekly_callback}
    \end{subfigure}

    \begin{floatnotes}
This figure reports Rambachan and Roth (2023) sensitivity intervals for the first monthly/weekly post-AI-release effect on callback rates. The specification includes bidder fixed effects, month fixed effects, and controls for normalized wage bid, reputation score rank percentile, and log job-level bid count. The monthly specification in addition controls for post-GPT4 $\times$ Access interaction. The weekly specification limits the pre-period to the four weeks between the release of GPT-4 and the release of the platform-native tool.
\end{floatnotes}
\end{figurebody}
\end{appfigure}
\clearpage

\begin{appfigure}
\centering
\begin{figurebody}{\textwidth}
\captionof{figure}{Partial $R^2$ of Cover Letter Tailoring by AI Tool Usage}
\label{fig:partialr2_ai_nonai}

      \centering 

  \includegraphics[width=\linewidth]{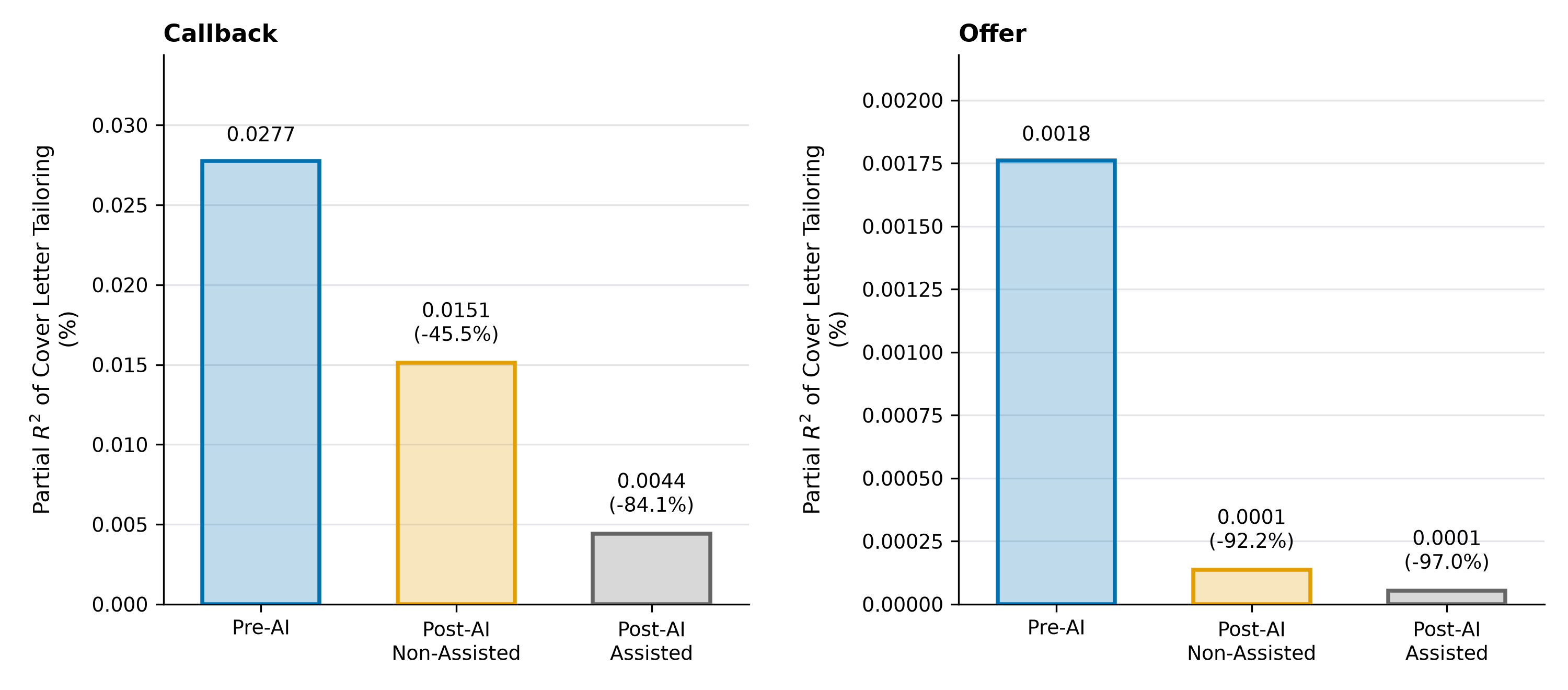}

    \begin{floatnotes}
The figure plots the partial $R^2$ of cover letter tailoring from bid-level regressions of callback and offer, estimated separately within three groups: bids submitted before the AI tool's release (Pre-AI), post-release bids that did not use the tool (Post-AI Non-Assisted), and post-release bids that did use the tool (Post-AI Assisted). Within each group, partial $R^2$ is computed by comparing a regression that includes tailoring, normalized wage bid, reputation score, and log bid count (with bidder and month fixed effects) to an otherwise identical regression that omits tailoring.
\end{floatnotes}
\end{figurebody}
\end{appfigure}

\begin{appfigure}
\centering
\begin{figurebody}{0.82\textwidth}
\captionof{figure}{Histogram of Mean Time Spent Editing}
\label{fig:hist_mean_5plus_two}

  \includegraphics[width=\linewidth]{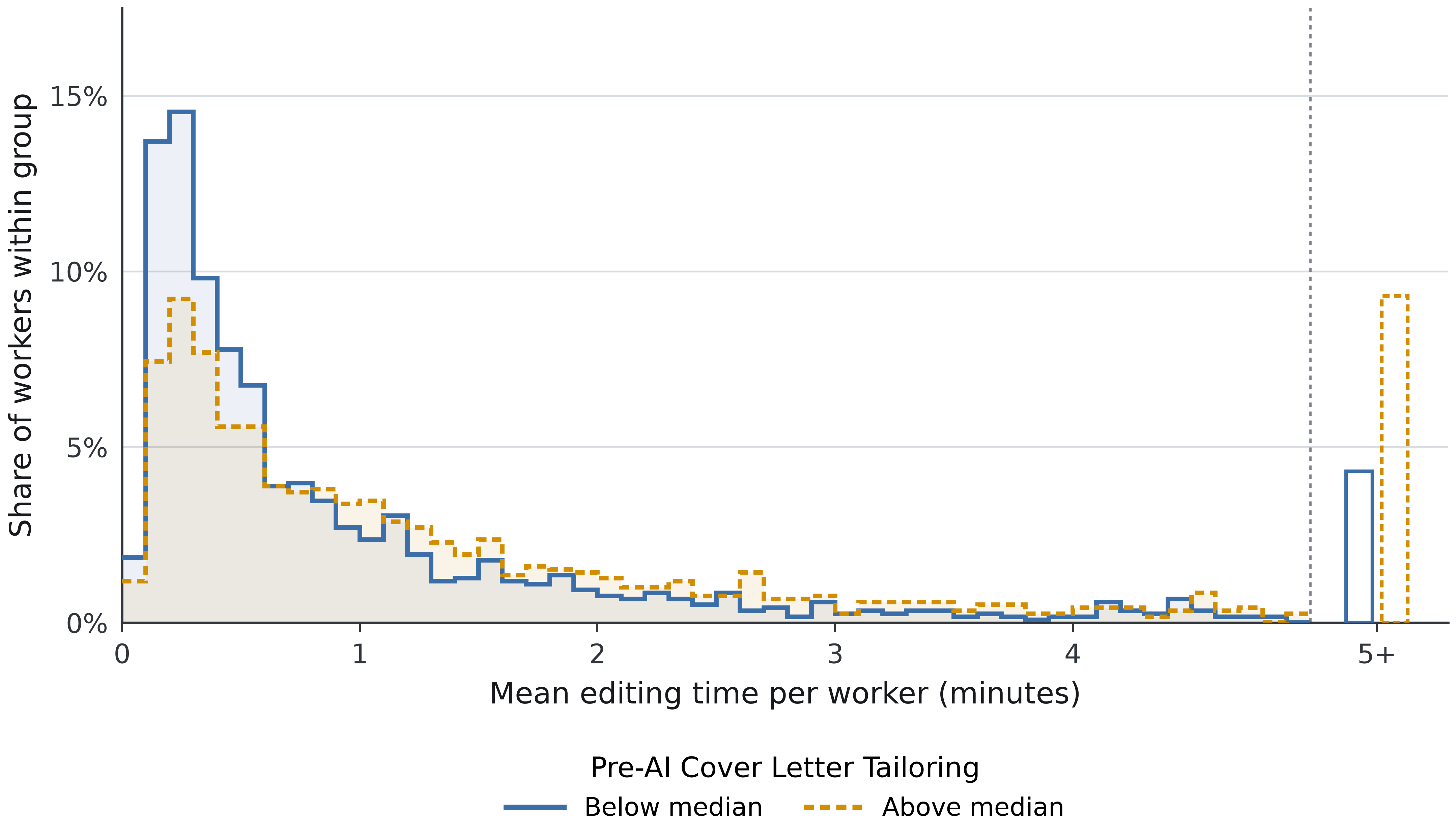}

  \begin{floatnotes}
Worker-level histogram. Each observation represents a worker’s average time, in minutes, between activating the tool and submitting a bid. Observations of 5 minutes or more are grouped in the final bin labeled ``5+''.
\end{floatnotes}
\end{figurebody}
\end{appfigure}

\clearpage
\newpage

\section{Econometric Decomposition \label{appendix_sec: decompose}}

Our setting violates the Stable Unit Treatment Value Assumption (SUTVA): The tool affects workers not only directly, by altering the cover letters of those who use it and, consequently, their market outcomes, but also indirectly through market-wide effects, by changing the overall salience of cover letter quality in the market and the competition intensity. To clarify what the estimated coefficients are capturing, consider the following econometric model.

Define the potential outcome $Y_{it}(d,m)$ as the outcome for worker $i$ in period $t$ with treatment status $d\in\{0,1\}$ and market status $m\in\{m_0, m_1\}$. The market status $m$ captures general features such as the salience of cover letter tailoring in firms’ hiring decisions or overall propensity of firms to hire. Let $A_i=1$ represent that worker $i$ belongs to the treatment group and $A_i=0$ otherwise. 

The two-way fixed effects regression identifies \[\mathbb{E}[Y_{i}(1,m_1)-Y_{i}(0,m_0)\mid A_i=1] - \mathbb{E}[Y_{i}(0,m_1)-Y_{i}(0,m_0)\mid A_i=0].\]

We can express that object through the following decomposition:
\begin{gather*}
    \underbrace{\mathbb{E}[Y_{i}(1,m_1)-Y_{i}(0,m_1)\mid A_i=1]}_{\text{ITT given post-AI market}} \quad + \\
    \underbrace{\mathbb{E}[Y_{i}(0,m_1)-Y_{i}(0,m_0)\mid A_i=1]-\mathbb{E}[Y_{i}(0,m_1)-Y_{i}(0,m_0)\mid A_i=0]}_{\text{Difference in market-wide spillovers}} .
\end{gather*}

The first term is our object of interest: the effect of the platform-native tool on workers’ outcomes while holding market conditions constant at the post-AI market. The second term reflects aggregate market changes that coincided with the introduction of the tool. Note that this term affects the results only insofar as aggregate shifts differentially impact the control and treatment groups. Treatment and control workers wrote similarly tailored cover letters and experienced similar callback rates in the pre-period. This makes it plausible that they are similarly exposed to market-level spillovers after the tool's introduction, suggesting that our difference-in-differences estimate is potentially close to the own-access effect holding market conditions at their post-AI level.

\end{document}